\documentclass[twoside,final,11pt]{entics}
\usepackage{enticsmacro}

\newcommand{\acknowledgements}[1]{%
  \par\vspace{2ex}\noindent\textit{Acknowledgements.} #1\par
}

\newcommand{\funding}[1]{%
  \par\vspace{1ex}\noindent\textit{Funding.} #1\par
}
\usepackage{comment}
\usepackage[all,2cell]{xy}

\usepackage{amsmath,amssymb}

\usepackage{cleveref}

\usepackage{graphicx}
\usepackage{multirow}
\usepackage{amsfonts}
\usepackage{mathrsfs}
\usepackage{xcolor}
\usepackage{manyfoot}
\usepackage{booktabs}
\usepackage{algorithm}
\usepackage{algorithmicx} 
\usepackage{algpseudocode}
\usepackage{listings}

\usepackage{quiver}
\usepackage{mathtools}
\usepackage{todonotes}

\def\comp{;}
\newcommand{\cobang}{\mathord{\vcenter{\hbox{\rotatebox[origin=c]{180}{\normalfont!}}}}}

\newcommand{\freccia}[3]{#2 \colon #1  \to #3}

\newcommand{\duefreccia}[3]{\xymatrix@C=0.5cm{#2 \colon #1  \ar@{=>}[r] &  #3}}

\newcommand{\comsquare}[8]{ \xymatrix@+1pc{ 
#1 \ar[r]^{#5} \ar[d]_{#6} & #2 \ar[d]^{#7} \\
#3 \ar[r]_{#8} & #4 
}}
\newcommand{\pullback}[8]{ \xymatrix@+1pc{ 
#1 \pullbackcorner \ar[r]^{#5} \ar[d]_{#6} & #2 \ar[d]^{#7} \\
#3 \ar[r]_{#8} & #4 
}}
\newcommand{\quadratocomm}[8]{ \xymatrix@+1pc{ 
#1 \ar[r]^{#5} \ar[d]_{#6} & #2 \ar[d]^{#7} \\
#3 \ar[r]_{#8} & #4 
}}
\newcommand{\comsquarelargo}[8]{ \xymatrix@+1pc{ 
#1 \ar[rr]^{#5} \ar[d]_{#6} && #2 \ar[d]^{#7} \\
#3 \ar[rr]_{#8} && #4 
}}
\newcommand{\parallelmorphisms}[4]{\xymatrix@+1pc{
#1 \ar @<+4pt>[r]^{#2} \ar @<-4pt>[r]_{#3} & #4
}}
\newcommand{\relation}[4]{\xymatrix@+1pc{
\angbr{#2}{#3}\colon #1 \ar @<+4pt>[r] \ar @<-4pt>[r] & #4
}}
\newcommand{\frecceparalleleopposte}[4]{\xymatrix@+1pc{
#1 \ar@<+4pt>[r]^{#2} \ar@<-4pt>@{<-}[r]_{#3} & #4
}}
\newcommand{\equalizer}[6]{\xymatrix@+1pc{
#1 \ar[r]^{#2} & #3 \ar @<+4pt>[r]^{#4} \ar @<-4pt>[r]_{#5} & #6
}}
\newcommand{\coequalizer}[6]{\xymatrix@+1pc{
 #1 \ar @<+4pt>[r]^{#2} \ar @<-4pt>[r]_{#3} & #4 \ar[r]^{#5} & #6
}}

\newcommand{\subobject}[3]{\xymatrix{
#1 \ar@{>->}[r]^{#2} & #3
}}

\newcommand{\pullbackcorner}[1][ul]{\save*!/#1+1.2pc/#1:(1,-1)@^{|-}\restore}

\def\mC{\mathcal{C}}

\def\mD{\mathcal{D}}

\def\Rel{\mathbf{Rel}}

\def\id{\operatorname{ id}}         
\def\dom{\operatorname{ dom}}      
\def\mass{\operatorname{ mass}}      

\newcommand{\angbr}[2]{\langle #1,#2 \rangle}

\usepackage{tikz}
\usepackage{tikzit}

\usetikzlibrary{shapes}

\usepackage{freetikz}
\usetikzlibrary{decorations.markings,positioning,patterns}
\usetikzlibrary{shadows}
\usepackage{array}
\usepackage{tikz-cd}
\usepackage{circuitikz}

\tikzstyle{nodonero}=[fill=black, draw=black, shape=circle]
\tikzstyle{box}=[fill=white, draw=black, shape=rectangle]
\tikzstyle{medium box}=[fill=white, draw=black, shape=rectangle, minimum width=0.7cm, minimum height=0.7cm]
\tikzstyle{bn}=[fill=black, draw=black, shape=circle, inner sep=1.5pt]
\tikzstyle{state}=[fill=white, draw=black, regular polygon, regular polygon sides=3, minimum width=0.8cm, shape border rotate=180, inner sep=0pt]
\tikzstyle{costate}=[fill=white, draw=black, regular polygon, regular polygon sides=3, minimum width=0.8cm, inner sep=0pt]
\tikzstyle{comp}=[fill={rgb,255: red,191; green,0; blue,64}, draw={rgb,255: red,191; green,0; blue,64}, shape=circle, inner sep=1.5pt]

\tikzstyle{ds}=[-, dashed, dash pattern=on 1mm off 1mm]

\usepackage[all,2cell]{xy}
\UseAllTwocells
\xyoption{v2}

\def\conf{MFPS 2026} 	
\volume{6}			
\def\volu{6}			
\def\papno{4}			
\def\mydoi{18722}
\def\lastname{Cioffo, Gadducci, Trotta} 
\def\runtitle{Between Markov and Restriction\dots} 
\def\copynames{C. Cioffo, F. Gadducci, D. Trotta}    

\begin{document}

\begin{frontmatter}

\title{Between Markov and Restriction:\\ Two More Monads on Categories for Relations}

\author{Cipriano Junior Cioffo}
\address{Department of Computer Science, University of Pisa}

\author{Fabio Gadducci}
\address{Department of Computer Science, University of Pisa}

\author{Davide Trotta}
\address{Department of Mathematics, University of Padova}

\begin{abstract}
The study of categories abstracting the structural properties of relations has been extensively developed over the years, resulting in a rich and diverse body of work.
In a previous paper we offered a survey providing a modern presentation of these ``categories for relations'' as instances of gs-monoidal categories,
showing how they arise as Kleisli categories of suitable symmetric monoidal monads. 
The end result was a taxonomy that organised numerous related concepts in the literature, including in particular 
Markov and restriction categories.
This paper further  enriches the taxonomy: it proposes two categories that are once more instances of gs-monoidal categories,
yet more abstract than Markov and restriction categories. They are characterised by an axiomatic notion of mass and domain of an arrow,
 the latter one of the key ingredients of restriction categories, which generalises the domain of partial functions. 
The paper then introduces mass and domain preserving monads, proving that the associated Kleisli categories in fact preserve the corresponding equations
and that these monads arise naturally for the categories of semiring-weighted relations.
\end{abstract}


\end{frontmatter}

\acknowledgements{
We are indebted to Tobias Fritz for the careful reading of a preliminary version of this paper and for sharing relevant pointers to the literature.}

\funding{This research was partly funded by the Advanced Research + Invention Agency (ARIA) Safeguarded AI Programme
and by the EU through the MSCA SE project QCOMICAL (Grant ID 101182520).}

\section{Introduction}
In recent years, the study of categories abstracting the properties of relations has been extensively developed both in mathematics and computer science, making it difficult
to identify a single basic notion that captures all the relevant aspects of these categories. However, a large part of the literature is based on the notion of symmetric monoidal category, which abstracts operations such as the product of relations, a leading example of a monoidal product that is not cartesian.

\vspace{.1cm}
\noindent
\begin{minipage}[l]{.45\linewidth}
\hspace{.4cm}
 Taking this fact as a starting point, in~\cite{cioffogadduccitrottataxonomy} we proposed a taxonomy for four families of (possibly order-enriched) categories,  
spanning from symmetric monoidal to cartesian monoidal categories: a fragment of such taxonomy, including only those categories that we will consider 
in this paper, is illustrated on the right.
\end{minipage}
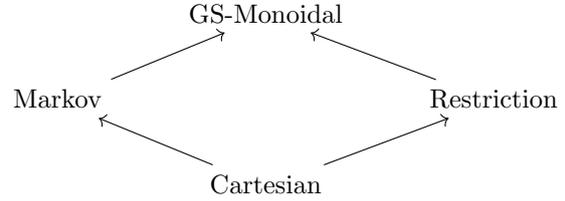
\begin{minipage}[r]{.20\linewidth}
\begin{tikzcd}
	& {\text{GS-monoidal}} \\
	{\text{Markov}} && {\hspace{-1cm}\text{Cartesian restriction}} \\
	& {\text{Cartesian monoidal}}
	\arrow[from=3-2, to=2-1]
	\arrow[from=3-2, to=2-3]
	\arrow[from=2-1, to=1-2]
	\arrow[from=2-3, to=1-2]
\end{tikzcd}
\end{minipage}
\vspace{.1cm}

Also known as copy/discard categories (shortly, CD categories)~\cite{cho_jacobs_2019}, \emph{garbage/share-monoidal categories} (shortly, gs-monoidal or GSM categories)~\cite{gadducci1996,CorradiniGadducci99,FritzGPT23} have structural 
arrows for duplicating 
and discharging objects: these families of arrows identify monoidal transformations, and make such categories models for relation-like structures.
If these arrows are natural, we obtain cartesian monoidal categories~\cite{Fox:CACC}.
%
If only the discharge arrows form a natural transformation, the resulting categories are known as Markov categories~\cite{Fritz_2020,cho_jacobs_2019}, which provide a framework for probabilistic reasoning~\cite{Fritz_2020}.
If only the duplicating arrows form a natural transformation, the resulting categories are known as cartesian restriction categories (or with restriction products)~\cite{Cockett02,Cockett03}, which provide a framework for partiality~\cite{Robinson88}.

The taxonomy included characterisation results for families of monads, e.g. showing that the Kleisli category of an affine monad on 
a Markov category is also Markov, or that the same holds for relevant monads on cartesian restriction categories, generalising the well-known fact that the Kleisli category 
of a commutative monad on a cartesian category is symmetric monoidal. Given the relevance of Kleisli categories for the theory of computation after 
the seminal work by Moggi~\cite{Moggi91}, it seemed useful 
to precisely state which axioms hold for a Kleisli category with respect to a given monad on a given base category.

This paper moves from \cite{cioffogadduccitrottataxonomy},
%
and its starting point lies in the serendipitous discovery that the Kleisli category of an affine monad on a cartesian restriction category is not necessarily a cartesian restriction one, 
yet it satisfies one key axiom of 
such categories, called the domain equation in the literature, which has recently come to the forefront for partial Markov categories~\cite{abs-2502-03477,LavoreR23}
and on its own for quasi-Markov categories, independently introduced in~\cite{mohammed2025partializationsmarkovcategories,empirical}.
%
%
We thus identified two classes of categories between gs-monoidal, Markov and cartesian restriction categories, 
sharing the same structural arrows, hence suitable as models of ``categories for relations''.
More precisely, we introduce \emph{mass} and \emph{domain categories}, 
and we show that both Markov and cartesian restriction categories are instances of such categories.
%

\vspace{.1cm}
\noindent
\begin{minipage}[l]{.30\linewidth}
\hspace{.4cm}
We thus obtain the diagram aside, which now includes mass, domain and weakly Markov categories. Furthermore, we prove that Markov categories are precisely those categories 
that are both mass and weakly Markov categories. Finally, we consider mass and domain preserving monads, showing under which conditions the associated Kleisli categories lift the structure of the base categories.
\end{minipage}
\begin{minipage}[r]{.20\linewidth}
\begin{tikzcd}
	& {\text{GS-monoidal}} \\
	\text{Weakly Markov} & \text{Mass} \\
	& \text{Domain} \\
	{\text{Markov}} && {\text{Cartesian restriction}} \\
	& {\text{Cartesian monoidal}}
	\arrow[from=2-1, to=1-2]
	\arrow[from=2-2, to=1-2]
	\arrow[from=3-2, to=2-2]
	\arrow[from=4-1, to=3-2]
	\arrow[from=4-1, to=2-1]
	\arrow[from=4-3, to=3-2]
	\arrow[from=5-2, to=4-1]
	\arrow[from=5-2, to=4-3]
\end{tikzcd}
\end{minipage}
\vspace{.1cm}


The paper has the following structure. 
\Cref{prem} recalls 
gs-monoidal, Markov and cartesian restriction categories, and it is rounded up with a few simple,
yet we believe original, results on monoid objects in gs-monoidal categories.
\Cref{presMassDom} introduces mass and domain categories, while \Cref{presFunctors} introduces mass and domain preserving functors. Then
the connections are shown with Markov and cartesian restriction categories in the former case, and with affine and relevant functors
in the latter.
\Cref{kleisli} contains the key results of our work, namely when either the mass or the domain structure is lifted to the Kleisli category,
and the decomposition of Markov categories (affine functors) in terms of weakly Markov and mass categories (weakly affine and mass preserving functors).
\Cref{WR} and \Cref{parMarkov} provide our case studies, showing how the notions we introduced are instantiated to the category of 
semiring-weighted relations and to recent work on partiality in Markov categories.
\Cref{oplax} concludes the paper with an analysis of order-enriched mass and domain categories and of their connections 
with oplax cartesian categories.


\section{Preliminaries}
\label{prem}
The first section recalls basic definitions about gs-monoidal, Markov and cartesian restriction categories.
The second section discusses monoid and comonoid objects in a symmetric monoidal category,
showing some simple results of which we are not aware of in the literature. 
We refer to \cite{cioffogadduccitrottataxonomy} and the references therein
for an introduction to gs-monoidal categories and their connection with categories for relations.

\subsection{GS-monoidal categories}

We fix a symmetric monoidal category $(\mC, \otimes, I)$. The axioms are presented using (upwards) string diagrams notation, 
and as usual in the graphical calculus for strict symmetric monoidal categories, the equations of string diagrams are understood modulo 
associativity of the monoidal product and cancellation of the unit.

\begin{definition}\label{dfn:gscategory}
    A \emph{gs-monoidal category} (GSM category for short) is a symmetric monoidal category $(\mC, \otimes, I)$ together with two distinguished arrows for every object $X$
\ctikzfig{copyanddel}
These arrows must be multiplicative with respect to the monoidal structure, meaning that they satisfy 
\ctikzfig{axiom_grb_cat}
\ctikzfig{comon-struct-mult-share-cat}
Also, every object $X$ has a cocommutative comonoid structure
    \ctikzfig{comonoid_share_cat_copy_del}
%
\end{definition}

We refer to $!_X : X \to I$ as the \emph{discharger} and to $\nabla_X : X \to X \otimes X$ as the \emph{duplicator}.

%

It is now well-known that
Markov and cartesian restriction categories are instances of GSM ones.

\begin{definition}\label{def:restriction}
Let $\mC$ be a GSM category. We say that  $\mC$ is a \emph{cartesian restriction category} if every arrow is \emph{copyable}, namely
    \ctikzfig{functional}

We say that $\mC$ is a \emph{Markov} category if every arrow is \emph{total}, namely
	\ctikzfig{full}

We say that $\mC$ is a \emph{cartesian monoidal} category if every arrow is copyable and total.
\end{definition}

As originally noted by Fox~\cite{Fox:CACC}, every cartesian monoidal category is in fact a cartesian category with a choice for the binary products and the terminal object.

 \subsection{Monoids in gs-monoidal categories}
 \label{monGS}
 
 We provide a few simple results for monoids in GSM categories. In the following left and right unitors are denoted by $\lambda$ and $\rho$, respectively{, and braidings by $\gamma$}.
 
\begin{lemma}
\label{cansem}
Let $\mC$ be a GSM category. Then each object $X$ has a canonical structure $\langle X, \Delta_X \rangle$ 
of a special semigroup object in $\mC$ for $\Delta_X = (\id_X \otimes !_X) ; \rho^{-1}_X$,
where special means that $\nabla_X ; \Delta_X = \id_X$.
Moreover, $I$ is a special commutative monoid object in $\mC$ for $\langle I, \Delta_I, \cobang_I \rangle$ 
with $\Delta_I = \rho^{-1}_I$ and $\cobang_I = \id_I$.
\end{lemma}
 We now investigate how these structures can be preserved along functors.
  
 \begin{proposition}
 \label{mon-pres}
Let $\mC$, $ {\mD}$ be symmetric monoidal categories, 
$\freccia{\mC}{F}{\mD}$ a lax symmetric monoidal functor with structural arrows
$\psi_{X,X}$ and $\psi_0$,
and $\langle X, \Delta_X, \cobang_X \rangle$ a (commutative) monoid object in $\mC$.
Then $\langle F(X), \Delta_{F(X)},  \cobang_{F(X)} \rangle$ is a (commutative) monoid object in $\mD$ with multiplication 
$\Delta_{F(X)} = \psi_{X,X} ; F(\Delta_X)$ and unit $\cobang_{F(X)} = \psi_0 ; F(\cobang_X)$.
\end{proposition}

The proposition above can be generalised for semigroup objects $\langle X, \Delta_X \rangle$. 

\begin{remark}
Note that in the case that $\mC$ and $ {\mD}$ above are GSM categories, the property for the canonical objects
introduced in Lemma~\ref{cansem} to be special is not preserved along $F$: see Lemma~\ref{specialLemma}.
\end{remark}

Finally, we show how the internal structure of an object is reflected on the hom-sets.

\begin{proposition}
\label{semihom}
Let $\mC$ be a GSM category. Then $\langle X, \Delta_X, \cobang_X \rangle$ is a (commutative) monoid object for $\mC$ if and only if for every object $Y$ in $\mC$ the pair 
$\langle  \mC(Y,X), m_{Y,X}, e_{Y,X} \rangle$
is a (commutative) monoid with multiplication $m_{Y,X}:\mC(Y,X)\times \mC(Y,X)\to \mC(Y,X)$ and unit $e_{Y,X} \in \mC(Y,X)$
satisfying 
\[
m_{Y,X}(f,g)=\nabla_Y\comp (f\otimes g) \comp m_{X\otimes X,X}((\id_X\otimes !_X){;\rho_x^{-1}},(!_X\otimes \id_X){;\lambda_x^{-1}})
\qquad
e_{Y,X} = !_Y\comp\cobang_X
\]
\end{proposition}
\begin{proof}
Assume that $\langle X, \Delta_X, \cobang_X \rangle$ is a monoid object for $\mC$. Then for every object $Y$ of $\mC$, the hom-set $\mC(Y,X)$ has a monoidal structure 
given by the following arrows
\begin{itemize}
	\item multiplication $m_{Y,X}:\mC(Y,X)\times \mC(Y,X)\to \mC(Y,X)$ as the assignment $(f,g)\mapsto \nabla_Y \comp (f\otimes g)\comp \Delta_X$;
	\item unit $e_{Y,X} \in \mC(Y,X)$ as the element $!_Y\comp \cobang_X$.
\end{itemize}
The associativity of $m_{Y,X}$ follows from the associativity of $\Delta_X$ and $\nabla_Y$. The unitality of $m_{Y,X}$ follows from the unitality of $\cobang_X$ and the axiom $\nabla_Y\comp (\id_Y\otimes !_Y);\rho_Y^{-1}=\id_Y$: for every $f\in \mC(Y,X)$ we have that 
\[\nabla_Y \comp (f\otimes (!_Y\comp \cobang_X))\comp \Delta_X= \rho_Y;(f\otimes \cobang_X)\comp \Delta_X=\rho_Y;(f\otimes \id_I);\rho_X^{-1}= f.\]
Moreover, since the monoid $\langle X, \Delta_X, \cobang_X\rangle$ is commutative, therefore it immediately follows that $m_{Y,X}$ is commutative for every $X,Y\in\mC$.
%


Conversely, let us define $(X,\Delta_X,\cobang_X)$ as $\Delta_X = m_{X\otimes X,X}((\id_X\otimes !_X){;\rho_x^{-1}},(!_X\otimes \id_X){;\lambda_x^{-1}})$ and $\cobang_X= e_{I,X}$. The unitality condition follows from the following computation
\begin{align}
	\id_X&=m_{X,X}(e_{X,X},\id_X)\tag{Unit.\ of $m$}\\
	&= \nabla_X;(e_{X,X}\otimes \id_X);m_{X\otimes X,X}(\id_X\otimes !_X,!_X\otimes \id_X) \tag{Def.\ of $m$}\\
	&= \nabla_X;((!_X;e_{I,X})\otimes (\id_X));m_{X\otimes X,X}((\id_X\otimes !_X){;\rho_x^{-1}},(!_X\otimes \id_X){;\lambda_x^{-1}}) \tag{Def.\ of $e_{X,X}$}\\
	&= \lambda_X; (e_{I,X}\otimes \id_X);m_{X\otimes X,X}((\id_X\otimes !_X){;\rho_x^{-1}},(!_X\otimes \id_X){;\lambda_x^{-1}}) \tag{Unit.\ of $\nabla_X$}\\
	&= \lambda_X; \cobang_X\otimes \id_X;\Delta_X \tag{Def.\ of $e_{I,X}$}
\end{align}
Associativity follows similarly.  Moreover, if $m_{Y,X}$ is commutative for every $X,Y\in\mC$, then the following computation shows that $(X,\Delta_X,\cobang_X)$ is commutative
\begin{align*}
	\Delta_X&= m_{X\otimes X,X}((\id_X\otimes !_X){;\rho_x^{-1}},(!_X\otimes \id_X){;\lambda_x^{-1}}) \tag{Def.\ of $\Delta_X$}\\
	&=m_{X\otimes X,X}((!_X\otimes \id_X){;\lambda_x^{-1}},(\id_X\otimes !_X){;\rho_x^{-1}})\tag{Comm.\ of $m$} \\
	&= \nabla_{X\otimes X};(((!_X\otimes \id_X);\lambda_x^{-1})\otimes ((\id_X\otimes !_X);\rho_x^{-1}));m_{X\otimes X,X}((\id_X\otimes !_X){;\rho_x^{-1}},(!_X\otimes \id_X){;\lambda_x^{-1}})  \tag{Def.\ of $m$}\\
	&= \gamma_{X,X};m_{X\otimes X,X}((\id_X\otimes !_X){;\rho_x^{-1}},(!_X\otimes \id_X){;\lambda_x^{-1}}) \tag{Def. \ref{dfn:gscategory}}\\
	&= \gamma_{X,X};\Delta_X.
\end{align*}
\end{proof}

As before, the result above can be generalised to semigroup objects $\langle X, \Delta_X \rangle$. 

\section{Mass and domain categories}
\label{presMassDom}
We introduce the two kinds of categories we focus our attention on. First, we recall the notions of mass and domain of an arrow, 
the latter generalising the domain of a function.

\begin{definition}
	Let $\mC$ be a gs-monoidal category and $\freccia{X}{f}{Y}$ an arrow in $\mC$. We define the \emph{mass} and the \emph{domain} of $f$ as the arrows below
\ctikzfig{domain}
\end{definition}

Textually, $\dom(f) = \nabla_X ; (\id_X \otimes f; !_Y) ; \rho^{-1}_X: X \to X$
and $\mass(f) = f; !_Y: X \to I$.

\begin{definition}\label{pr_cat}
A GSM category $\mC$ is 
a \emph{domain preserving category} (shortly, \emph{domain category}) if for every arrow $f:X\to Y$ in $\mC$ the equality on the left below holds, 
while it is a \emph{mass preserving category} (shortly, \emph{mass category}) if for every arrow $f:X\to Y$ in $\mC$ the equality on the right below holds
\begin{center}
\begin{tikzpicture}
	\begin{pgfonlayer}{nodelayer}
		\node [style=none] (0) at (2.5, -1.5) {$X$};
		\node [style=bn] (1) at (3, -0.5) {};
		\node [style=none] (2) at (2, 0.5) {};
		\node [style=none] (3) at (4, 0.5) {};
		\node [style=none] (4) at (3, -1.5) {};
		\node [style=box] (6) at (4, 1.5) {$f$};
		\node [style=none] (7) at (2, 3) {};
		\node [style=none] (8) at (4, 3) {};
		\node [style=none] (9) at (1.5, 3) {$Y$};
		\node [style=bn] (11) at (4, 3) {};
		\node [style=box] (13) at (2, 1.5) {$f$};
		\node [style=none] (14) at (5.5, 1) {$=$};
		\node [style=none] (15) at (7, -1.5) {};
		\node [style=none] (16) at (7, 3) {};
		\node [style=none] (17) at (6.5, 3) {$Y$};
		\node [style=box] (18) at (7, 1.5) {$f$};
		\node [style=none] (19) at (6.5, -1.5) {$X$};
	\end{pgfonlayer}
	\begin{pgfonlayer}{edgelayer}
		\draw [in=165, out=-90] (2.center) to (1);
		\draw [in=15, out=-90] (3.center) to (1);
		\draw (4.center) to (1);
		\draw (7.center) to (2.center);
		\draw (8.center) to (3.center);
		\draw (16.center) to (15.center);
	\end{pgfonlayer}
\end{tikzpicture}
\qquad
\begin{tikzpicture}
	\begin{pgfonlayer}{nodelayer}
		\node [style=none] (0) at (2.5, -1.5) {$X$};
		\node [style=bn] (1) at (3, -0.5) {};
		\node [style=none] (2) at (2, 0.5) {};
		\node [style=none] (3) at (4, 0.5) {};
		\node [style=none] (4) at (3, -1.5) {};
		\node [style=box] (6) at (4, 1.5) {$f$};
		\node [style=none] (7) at (2, 3) {};
		\node [style=none] (8) at (4, 3) {};
		\node [style=bn] (11) at (4, 3) {};
		\node [style=box] (13) at (2, 1.5) {$f$};
		\node [style=none] (14) at (5.5, 1) {$=$};
		\node [style=none] (15) at (7, -1.5) {};
		\node [style=none] (16) at (7, 3) {};
		\node [style=box] (18) at (7, 1.5) {$f$};
		\node [style=none] (19) at (6.5, -1.5) {$X$};
		\node [style=bn] (20) at (7, 3) {};
		\node [style=bn] (21) at (2, 3) {};
	\end{pgfonlayer}
	\begin{pgfonlayer}{edgelayer}
		\draw [in=165, out=-90] (2.center) to (1);
		\draw [in=15, out=-90] (3.center) to (1);
		\draw (4.center) to (1);
		\draw (7.center) to (2.center);
		\draw (8.center) to (3.center);
		\draw (16.center) to (15.center);
	\end{pgfonlayer}
\end{tikzpicture}
\end{center}
\end{definition}

\begin{remark}
The left-most equation  above can be written as  $\dom(f);f=f$, which is known as the domain equation
in the literature on restriction categories or quasi-totality in \cite{LavoreR23,abs-2502-03477}. 
In the literature on categorical probability, domain categories have been recently introduced as quasi-Markov 
categories~\cite{empirical,mohammed2025partializationsmarkovcategories}.
\end{remark}


\begin{remark}
The right-most equation above can be written as $\dom(f);\mass(f)=\mass(f)$, and the notion of mass is taken from weakly-Markov categories~\cite{FritzGPT23}. 
Clearly, domain categories are mass categories. 
The converse does not hold in general (see Section~\ref{WR}); 
however, as shown in~\cite[Lemma 3.17]{abs-2502-03477}, the
two conditions are equivalent in partial Markov categories.
\end{remark}

It is an easy check that the pre- and post-composition with structural arrows do not change the domain.

\begin{lemma}
Let $\mC$ be a domain category and $\freccia{X}{f}{Y}$ an arrow in $\mC$. Then it holds 
\begin{itemize}
 \item $\dom (f;\nabla_Y)=\dom (\nabla_X; f\otimes f) = \dom(f)$;
 \item $\dom (f;!_Y)=\dom (f)$.
\end{itemize}
\end{lemma}

\begin{lemma}\label{lemma:PR_subcategories}
Let $\mC$ be either a Markov or a cartesian restriction category. Then it is a domain category.
\end{lemma}
\begin{example}
The category $\Rel$ of sets and relations is the leading example of a GSM category (with respect to the ordinary product of sets, see \cite[Rem. 2.16]{FritzGCT23})
that is also a domain category, yet neither a Markov nor a restriction one. 
We will see how this generalises to  semiring-weighted relations in Section~\ref{WR}.
\end{example}

\subsection{More on restriction categories}
\label{appRes}

Cartesian restriction categories are an instance of the more general restriction categories. Adopted since~\cite{Cockett02} as a categorical abstraction of partiality,
their standard presentation is given below.

\begin{definition}\label{dfn: restriction category}
A \emph{restriction structure} on a category  $\mC$ is an assignment which sends every arrow $f:X\to Y$ of $\mC$ to an arrow $\overline{f}:X\to X$ 
such that the following conditions hold
\begin{description}
	\item[(R.1)] $f \circ \overline{f}=f$,
	\item[(R.2)] $\overline{f} \circ \overline{g} = \overline{g} \circ \overline{f}$ for $g: X \to W$,
	\item[(R.3)] $\overline{g \circ \overline{f}}= \overline{g} \circ \overline{f}$ for $g: X \to W$,
	\item[(R.4)] $\overline{g} \circ f= f\circ \overline{g \circ f}$ for $g: Y \to W$.
\end{description}
A \emph{restriction category} is a category equipped with a restriction structure.
\end{definition}

Note that (R.1) can be equivalently replaced by requiring $\overline{id_X} = id_X$ for every $X\in\mC$. 

A GSM category potentially supports the restriction structure, given as $\overline{f} = dom(f)$,
and as we noted (R.1) is precisely the axiom of domain categories. 
Using the restriction structure, the original definition of cartesian restriction categories was given 
in terms of restriction terminal object and products~\cite{Cockett07}. 
The correspondence with GMS categories with copyable arrows has been proved e.g. 
in~\cite[Proposition~2.58]{cioffogadduccitrottataxonomy}.

Here we recall the notion of positivity from the literature on Markov categories~\cite{Fritz_2020}
and we then prove a result showing that perhaps surprisingly mass categories already allow for a characterisation 
of partiality, as long as they satisfy the positivity condition.

\begin{definition}\label{positive}
    A GSM category is called \emph{positive} 
    if for every pair of arrows $f:X\to Y$ and $g: Y \to W$ such that 
    $ f; g$ is copyable then
\ctikzfig{positivity}
\end{definition}

\begin{proposition}
\label{positive as restriction}
	Let $\mC$ be a positive mass category. Then it is a restriction category.
\end{proposition}
\begin{proof}
	Define  $\overline{f} = \mathrm{dom}(f)$. The axioms $\overline{id_X} = id_X$ and ($R.2$) and ($R.3$) of Definition~\ref{dfn: restriction category} hold for any GSM category. 
	Since $\mC$ is a mass category, for every arrow $h:X\to W$ the composition $h;!_W$ is copyable. 
	Hence, applying positivity to $f$ and $g;!_W$ one obtains $(R.4)$ of Definition~\ref{dfn: restriction category}.
\end{proof}

\section{Functors preserving mass and domain}
\label{presFunctors}

Building on lax monoidal functors (see Appendix~\ref{sec:lax_app}), we now recall affine and relevant functors~\cite{Jacobs1994}.

\begin{definition}\label{def gsmonoidal functor}
Let $\mC$, $ {\mD}$ be GSM categories and $\freccia{\mC}{F}{\mD}$ a lax monoidal functor with structural arrows
$\psi_{X,Y}$ and $\psi_0$. We say that $F$ is \emph{affine}  if the diagram on the left commutes for all $X$ in $\mC$,
and that $F$ is  \emph{relevant} if the diagram on the right commutes for all $X$ in $\mC$

\[
\begin{tikzcd}[column sep=tiny]
	F(X) && {F(I)} \\
& I
\arrow["{F(!_X)}", from=1-1, to=1-3]
\arrow["{!_{F(X)}}"', from=1-1, to=2-2]
\arrow["{\psi_0}"', from=2-2, to=1-3]
\end{tikzcd}
\qquad
\qquad
	\begin{tikzcd}[column sep=tiny]
		{F(X)} && {F(X\otimes X)} \\
		& {F(X)\otimes F(X)}
		\arrow["{F(\nabla_X)}", from=1-1, to=1-3]
		\arrow["{\nabla_{FX}}"', from=1-1, to=2-2]
		\arrow["{\psi_{X,X}}"', from=2-2, to=1-3]
	\end{tikzcd}
	\]

We say that $F$ is \emph{cartesian monoidal} if it is both affine and relevant.
\end{definition}

We generalise these definitions in order to capture domain and mass categories.

\begin{definition}
Let $\mC$, $ {\mD}$ be GSM categories and 
$\freccia{\mC}{F}{\mD}$ a lax symmetric monoidal functor with structural arrows
$\psi_{X,Y}$ and $\psi_0$. We say that $F$ is \emph{domain preserving}  if the left-most diagram below 
commutes for all $X$ and that $F$ is \emph{mass preserving} if the right-most diagram below commutes for all $X$
\begin{center}
\begin{tikzcd}
	{F(X)\otimes F(X)} & {F(X)} \\
	{F(X \otimes X)} & {F(X\otimes I)}
	\arrow["{\psi_{X,X}}"', from=1-1, to=2-1]
	\arrow["{\nabla_{F(X)}}"', from=1-2, to=1-1]
	\arrow["{F(\rho_X)}", from=1-2, to=2-2]
	\arrow["{F(\id_X \otimes !_X)}"', from=2-1, to=2-2]
\end{tikzcd} \quad \begin{tikzcd}
	{F(X)\otimes F(X)} & {F(X)} & F(I) \\
	{F(X \otimes X)} & & F(I \otimes I)
	\arrow["{\psi_{X,X}}"', from=1-1, to=2-1]
	\arrow["{\nabla_{F(X)}}"', from=1-2, to=1-1]
	\arrow["{F(!_X)}", from=1-2, to=1-3]
	\arrow["{F(!_X \otimes !_X)}"', from=2-1, to=2-3]
	\arrow["{F(\rho_{I})}", from=1-3, to=2-3]
\end{tikzcd} 
\end{center}
\end{definition}

Note that the two diagrams above coincide for $X = I$, and we say that  
$F$ is \emph{unital domain preserving}  if they commute only for $X = I$.

\begin{proposition}
Let $\mC$ and $ {\mD}$ be GSM categories and 
$\freccia{\mC}{F}{\mD}$ a lax symmetric monoidal functor. 
If $F$ is domain preserving then it is mass preserving.
\end{proposition}
\begin{proof}
Immediate, since the right-most diagram below commutes for any monoidal category and lax symmetric monoidal functor
\begin{center}
\begin{tikzcd}
	{F(X)\otimes F(X)} & {F(X)} & F(I) \\
	{F(X \otimes X)} & {F(X\otimes I)} & F(I \otimes I)
	\arrow["{\psi_{X,X}}"', from=1-1, to=2-1]
	\arrow["{\nabla_{F(X)}}"', from=1-2, to=1-1]
	\arrow["{F(!_X)}", from=1-2, to=1-3]
	\arrow["{F(\id_X \otimes !_X)}"', from=2-1, to=2-2]
	\arrow["{F(!_X \otimes \id_I)}"', from=2-2, to=2-3]
	\arrow["{F(\rho_{I})}", from=1-3, to=2-3]
	\arrow["{F(\rho_{X})}", from=1-2, to=2-2]
\end{tikzcd}
\end{center}
%
\end{proof}

Lemma~\ref{cansem} states that for a GSM category $\mC$ every object $X$ is canonically a special semigroup object $\langle X, \Delta_X \rangle$ in  $\mC$
with $\Delta_X = (\id_X \otimes !_X) ; \rho^{-1}_X$, and Proposition~\ref{mon-pres} tells that, for a lax symmetric monoidal functor $F$, 
$\langle F(X), \psi_{X,X} ; F(\Delta_X) \rangle$ is a semigroup object in  $\mD$. Hence, we obtain the characterisation below.

\begin{lemma}
\label{specialLemma}
Let $\mC$ and $ {\mD}$ be GSM categories and 
$\freccia{\mC}{F}{\mD}$ a lax symmetric monoidal functor with structural arrows $\psi_{X,Y}$ and $\psi_0$. Then it holds that
\begin{itemize}
	\item $F$ is domain preserving if and only if $\langle F(X), \psi_{X,X} ; F(\Delta_X) \rangle$ is a special semigroup object in $\mD$ for every object $X$ in $\mC$; 							
	\item $F$ is unital domain preserving if and only if $\langle F(I), \psi_{I,I} ; F(\rho^{-1}_I), \psi_0 \rangle$ is a special commutative monoid object in $\mD$.
\end{itemize}
\end{lemma}

The second item above can be further strengthened.

\begin{proposition}\label{lemma:mass implies unital domain}
Let $\mC$, $ {\mD}$ be GSM categories and 
$\freccia{\mC}{F}{\mD}$ a lax symmetric monoidal functor. 
If $\mD$ is a cartesian restriction category then 
$F$ is mass preserving if and only if it is unital domain preserving.
\end{proposition}
\begin{proof}
If $\mD$ is a cartesian restriction category, we can take the definition of unital domain preserving, which states that $\nabla_{F(I)} ; \psi_{I,I} = F(\rho_I)$, since $\id_I = !_I$, and consider the definition of mass preserving which corresponds to the outer diagram below
\[
\begin{tikzcd}
	{F(X)\otimes F(X)} & {F(X)} & {F(I)} \\
	&& {F(I)\otimes F(I)} \\
	{F(X\otimes X)} && {F(I\otimes I)}
	\arrow["{F(!_X)\otimes F(!_X)}"', from=1-1, to=2-3]
	\arrow["{\psi_{X,X}}"', from=1-1, to=3-1]
	\arrow["{\nabla_{F(X)}}"', from=1-2, to=1-1]
	\arrow["{F(!_I)}", from=1-2, to=1-3]
	\arrow["{\nabla_{F(I)}}", from=1-3, to=2-3]
	\arrow["{F(\rho_I)}", curve={height=-45pt}, from=1-3, to=3-3]
	\arrow["{\psi_{I,I}}", from=2-3, to=3-3]
	\arrow["{F(!_X\otimes !_X)}"', from=3-1, to=3-3]
\end{tikzcd}\]
Since the inner bottom diagram commutes by naturality of $\psi$ and the inner top one by definition of restriction category, the outer diagram commutes as well. Hence, $F$ is mass preserving.
\end{proof}

\begin{proposition}
Let $F$ be either an affine or a relevant functor. Then it is domain preserving.
	\end{proposition}
\begin{proof}
For relevant functors, consider the diagram below
\[
\begin{tikzcd}
	{F(X)\otimes F(X)} & {F(X)} \\
	{F(X \otimes X)} & {F(X\otimes I)}
	\arrow["{F(\nabla_X)}"', from=1-2, to=2-1]
	\arrow["{\psi_{X,X}}"', from=1-1, to=2-1]
	\arrow["{\nabla_{F(X)}}"', from=1-2, to=1-1]
	\arrow["{F(\rho)}", from=1-2, to=2-2]
	\arrow["{F(\id_X \otimes !_X)}"', from=2-1, to=2-2]
\end{tikzcd}\]
The upper diagram commutes since it is the property of being relevant and the lower diagram is just the image of the comonoid equation for $X$.

For affine functors, consider instead the diagram below
\[
\begin{tikzcd}
	& {F(X)\otimes F(X)} & {F(X)} \\
	{F(X) \otimes I} & {F(X) \otimes F(I)} & {F(X\otimes I)}
	\arrow["{\id \otimes !_{F(X)}}"', from=1-2, to=2-1]
	\arrow["{\id \otimes \psi_0}"', from=2-1, to=2-2]
	\arrow["{\id \otimes F(!_X)}", from=1-2, to=2-2]
	\arrow["{\nabla_{F(X)}}"', from=1-3, to=1-2]
	\arrow["{F(\rho)}", from=1-3, to=2-3]
	\arrow["{\psi_{X,I}}"', from=2-2, to=2-3]
\end{tikzcd}\]
The right-most diagram is just the domain preserving equation,
the left-most diagram is the property of being affine, and the 
outermost diagram is the unitality equation for $\rho_X$, since $\nabla_{F(X)} ; (\id \otimes !_{F(X)}) = \rho_{F(X)}$
is the codomain equation for $F(X)$. Hence, also the right-most diagram commutes.
\end{proof}


\section{Looking at Kleisli categories}
\label{kleisli}
It is well-known that a symmetric monoidal monad $T$ (see Appendix~\ref{sec:lax_app}) on a Markov category $\mC$ is affine if and only if the Kleisli category $\mC_T$ is again a Markov category and the same occurs for relevant monads with respect to cartesian restriction categories (see \cite[Theorem 4.3]{Jacobs1994}).
Now, let us say that a symmetric monoidal monad is domain/unital domain/mass preserving if the underlying functor is so.
We can generalise the result above for cartesian restriction (hence also for cartesian) categories.

We first provide a simple, yet lengthy, technical lemma. 

\begin{lemma}\label{lemma dom f;f nelle kleisli}
Let $\mC$ be a cartesian restriction category and $T$ a symmetric monoidal monad on it. For every arrow $\freccia{X}{f}{Y}$ of $\mC_T$ the domain 
equation in $\mC_T$ is
\[
\dom(f);^{\sharp}f=
f;\nabla_{TY};(T(\id_Y)\otimes T(!_Y));c_{Y,I};T(\rho^{-1})
\]
\end{lemma}

\begin{proof}
		Let $f:X\to Y$ be an arrow in $\mC_T$, i.e. an arrow $f:X\to T(Y)$ in $\mC$. Consider the composition
		\[\nabla_X^{\sharp};^{\sharp} (f\otimes^{\sharp}  f);^{\sharp}(\id\otimes^{\sharp}!^{\sharp}_Y)\]
		in $\mC_T$ that is equal to 
		\[\nabla_X; \eta_{X\otimes X}; T(f\otimes f);T( c_{Y,Y});\mu_{Y\otimes Y}; T(\eta_Y\otimes (!_Y;\eta_I));T( c_{Y,I});\mu_{Y\otimes I}\]
		By naturality of  $\eta$, we have that this composition is equal to
				\[ \nabla_X; (f\otimes f);\eta_{TY\otimes TY};
		T( c_{Y,Y});\mu_{Y\otimes Y}; T(\eta_Y\otimes (!_Y;\eta_I));T( c_{Y,I});\mu_{Y\otimes I}
		\]
		and this is equal to 
				\[\nabla_X;	(f\otimes f); c_{Y,Y};\eta_{T(Y\otimes Y)};
		\mu_{Y\otimes Y};T(\eta_Y\otimes (!_Y;\eta_I));T( c_{Y,I});	\mu_{Y\otimes I}
		\]
		that is 
			\[\nabla_X;	(f\otimes f); c_{Y,Y};T(\eta_Y\otimes (!_Y;\eta_I));T( c_{Y,I});\mu_{Y\otimes I}
		\]
		This is equal to
				\[\nabla_X;	(f\otimes f); c_{Y,Y};T(\id_Y \otimes !_Y);
		T( (\eta_Y\otimes \eta_I);c_{Y,I});\mu_{Y\otimes I}
		\]
		Now, since $ (\eta_Y\otimes \eta_I); c_{Y,I}= \eta_{Y\otimes I}$ (because the monad is symmetric monoidal), this is equal to
				\[\nabla_X;	(f\otimes f); c_{Y,Y};T(\id_Y \otimes !_Y)
		\]
		but now $c_{Y,Y};T(\id_Y \otimes !_Y)=(T(\id_T)\otimes T(!_Y));c_{Y,I}$, hence
				\[\nabla_X;(f\otimes f);	(T(\id_Y)\otimes T(!_Y));
		c_{Y,I}
		\]
		that is
		\[f;\nabla_{TY};(T(\id_Y)\otimes T(!_Y));c_{Y,I}	\]		
		Hence, we have proved that 
			\[\nabla_X^{\sharp};^{\sharp} (f\otimes^{\sharp}  f);^{\sharp}(\id\otimes^{\sharp}!^{\sharp}_Y)=f;\nabla_{TY};(T(\id_Y)\otimes T(!_Y));c_{Y,I}\]
and hence, since $(\rho^{-1}_Y)^{\sharp}=\rho^{-1}_Y;\eta_{Y}$, we can conclude that
\[\dom(f);^{\sharp}f=\nabla_X^{\sharp};^{\sharp} (f\otimes^{\sharp}  f);^{\sharp}(\id\otimes^{\sharp}!^{\sharp}_Y);^{\sharp}(\rho_Y^{-1})^{\sharp}=f;\nabla_{TY};(T(\id_Y)\otimes T(!_Y));c_{Y,I};T(\rho_Y^{-1})\]
\end{proof}

We now move to the key result concerning domain preservation.

\begin{theorem}\label{thm: restriction + special implies kl in PR}
		Let $\mC$ be a cartesian restriction category and $T$ a symmetric monoidal monad on $\mC$. Then $T$ is domain preserving if and only if $\mC_T$ is a domain category. 
	\end{theorem}
\begin{proof}
		Let us consider an arrow $f:X\to Y$ in $\mC_T$.
		By Lemma \ref{lemma dom f;f nelle kleisli}, we have that
			\[\dom(f);^{\sharp}f=\nabla_X^{\sharp};^{\sharp} (f\otimes^{\sharp}  f);^{\sharp}(\id\otimes^{\sharp}!^{\sharp}_Y);^{\sharp}(\rho^{-1})^{\sharp}=f;\nabla_{TY};(T(\id_Y)\otimes T(!_Y));c_{Y,I};T(\rho^{-1}).\]
		If $T$ is domain preserving, then
		\[\nabla_{TY};(T(\id_Y)\otimes T(!_Y));c_{Y,I}= \nabla_{TY};c_{Y,Y};(T(\id_Y)\otimes T(!_Y))=T(\rho).\]
		Hence, we have that $\dom(f);^{\sharp}f=f$.

		Vice versa, if $\mC_T$ is a domain category, i.e.\ $\dom(f);^{\sharp}f=f$ for every arrow $f$, in particular, we will have that the equation holds for $\freccia{TY}{f=\id_{TY}}{Y}$ of $\mC_T$. But this means that 
		\[\dom(\id_{TY});^{\sharp}\id_{TY}= \nabla_{TY};(T(\id_Y)\otimes T(!_Y));c_{Y,I};T(\rho^{-1})=\id_{TY}\]
		and hence that 
		\[\nabla_{TY};(T(\id_Y)\otimes T(!_Y));c_{Y,I}=\nabla_{TY};c_{Y,Y};(T(\id_Y)\otimes T(!_Y))=T(\rho).\] 
		Then, we can conclude that the monad is domain preserving.
\end{proof}

The result above implies as a corollary the starting point of our investigation.

\begin{corollary}\label{cor:thm}
Let $\mC$ be a cartesian restriction category and $T$ an affine monad. Then $\mC_T$ is a domain category. 
\end{corollary}

Now we move to the key result concerning mass preservation.

\begin{theorem}
Let $\mC$ be a cartesian restriction category and  $T$ a symmetric monoidal monad on $\mC$. Then $T$ is mass preserving if and only if  $\mC_T$ is a mass category.
\end{theorem}
\begin{proof}
	By Lemma \ref{lemma dom f;f nelle kleisli}, we have that 
	\[\dom(f);^{\sharp}f=\nabla_X^{\sharp};^{\sharp} (f\otimes^{\sharp}  f);^{\sharp}(\id\otimes^{\sharp}!^{\sharp}_Y);^{\sharp}(\rho^{-1})^{\sharp}=f;\nabla_{TY};(T(\id_Y)\otimes T(!_Y));c_{Y,I};T(\rho^{-1})\]
and hence 
\[\dom(f);^{\sharp}f;^{\sharp}!^{\sharp}_Y=f;\nabla_{TY};(T(\id_Y)\otimes T(!_Y));c_{Y,I};T(\rho^{-1}); T(!_Y)\]

Now, if $T$ is mass preserving, then
		\[\nabla_{TY};(T(\id_Y)\otimes T(!_Y));c_{Y,I};  T(!_Y)=T(!_Y)\]
		Hence, we have that
		\[\dom(f);^{\sharp}f;^{\sharp}!_Y^{\sharp}=f;^{\sharp}!_Y^{\sharp}\]

		One can prove the converse by using the same argument as in Theorem \ref{thm: restriction + special implies kl in PR}.
\end{proof}

\begin{corollary}
	Let $\mC$ be a cartesian restriction category and $T$ a symmetric monoidal monad on $\mC$. 
	Then $T$ is unital domain preserving if and only if  $\mC_T$ is a mass category.
\end{corollary}

\subsection{Weakly Markov vs domain preservation}
A recent addition to the taxonomy surveyed in~\cite{cioffogadduccitrottataxonomy} are weakly Markov categories,  which are intermediate between Markov and GSM categories~\cite{FritzGPT23}. 
In this section we explain how such a notion interacts with the (unital) domain preservation property.
\begin{definition}\label{defweaklymarkov}
Let $\mC$ be a GSM category. We say that it is \emph{weakly Markov} if for every object $Y$, the commutative monoid 
$\langle  \mC(Y,I), \nabla_Y \comp ( - \otimes -)\comp \rho^{-1}_I, !_Y \rangle$ is a group. 
\end{definition}
\begin{remark}\label{rem: markov are weaklymarkov}
	A Markov category is weakly Markov: the commutative monoid $\langle  \mC(Y,I), \nabla_Y \comp ( - \otimes -)\comp \rho^{-1}_I, !_Y \rangle$ is trivial, since in a Markov category the hom-set
$\mC(Y,I)$ is the singleton for any object $Y$, see \cite{FritzGPT23}.
\end{remark}

In other words, there is a function $(-)^{-1}: \mC(Y,I) \to \mC(Y,I)$ satisfying the obvious equations.
Now, by Lemma~\ref{cansem}
we know that $\langle I, \rho^{-1}_I, \id_I \rangle$ is a commutative monoid object
in $\mC$ and by Proposition~\ref{mon-pres} that for a symmetric monoidal monad $T$
on $\mC$ the same holds for $\langle T(I), \psi_{I,I}; T(\rho^{-1}_I), \psi_0 \rangle$.

\begin{definition}\label{defweaklyaffine}
Let $\mC$ and $ {\mD}$ be GSM categories and 
$\freccia{\mC}{F}{\mD}$ a lax symmetric monoidal functor with structural arrows
$\psi_{X,Y}$ and $\psi_0$. We say that $F$ is \emph{weakly affine} if the commutative monoid object
$\langle F(I), \psi_{I,I}; F(\rho^{-1}_I), \psi_0 \rangle$ in $\mD$ is a group.
\end{definition}

For a monoid object $(X, \Delta_X , \cobang_X)$ being a group object in a GSM category means the existence of an inverse arrow $\iota_X: X \to X$ such that
$\nabla_X ; (\id_X \otimes \iota_X) ; \Delta_X = !_X ; \cobang_X$, which is to say that $\iota_X$ is the antipode of a Hopf monoid object $X$.

\begin{remark}
The observation above suggests that, for a lax symmetric monoidal functor $F$,
the condition of being weakly affine is equivalent to the commutativity of the diagram below for $X = I$.
\[
\begin{tikzcd}
	{F(X)\otimes F(X)} & {F(X)\otimes F(X)} & {F(X)} & I & F(I) \\
	{F(X \otimes X)} & & & & F(I \otimes I)
	\arrow["{\psi_{X,X}}"', from=1-1, to=2-1]
	\arrow["{\id_{F(X)} \otimes \iota}"', from=1-2, to=1-1]
	\arrow["{\nabla_{F(X)}}"', from=1-3, to=1-2]
	\arrow["{!_{F(X)}}", from=1-3, to=1-4]
	\arrow["{\psi_0}", from=1-4, to=1-5]
	\arrow["{F(\id_X \otimes !_X)}"', from=2-1, to=2-5]
	\arrow["{F(\rho_{I})}", from=1-5, to=2-5]
\end{tikzcd}\]
\end{remark}

We now restate the key property for weakly affine monads (see~\cite[Proposition~3.6]{FritzGPT23}).

\begin{proposition}
Let $\mC$ be a cartesian monoidal category and $T$ a symmetric monoidal monad on $\mC$. Then $T$ is weakly affine if and only if $\mC_T$ is weakly Markov.
\end{proposition}

However, in the framework of domain categories 
we can be much more nuanced and actually \emph{decompose} affine monads.
%

\begin{proposition}\label{prop:markov iff weakly affine and unital domain preserving}
Let $\mC$ and ${\mD}$ be GSM categories and 
$\freccia{\mC}{F}{\mD}$ a lax symmetric monoidal functor. 
If ${\mD}$ is a Markov category then
$F$ is affine if and only if it is weakly affine and unital domain preserving.
%
\end{proposition}

\begin{proof}
	The assumptions of weakly affine and unital domain preservation imply that, for the monoid $F(I)$ with  multiplication $m= \psi_{I,I};F(\rho_I^{-1})$ and the inverse arrow $\iota$, it holds
	\[\begin{tikzpicture}
	\begin{pgfonlayer}{nodelayer}
		\node [style=none] (14) at (-4, -3) {};
		\node [style=bn] (15) at (-4, 0) {};
		\node [style=none] (16) at (-4, 2.75) {};
		\node [style=none] (17) at (-5, -2) {$!_{F(I)}$};
		\node [style=none] (18) at (-4, -4) {};
		\node [style=none] (19) at (-4, -4) {$F(I)$};
		\node [style=none] (20) at (-4.5, 0.5) {};
		\node [style=none] (21) at (-4.5, 1.5) {$\psi_0$};
		\node [style=none] (23) at (-4, 3.75) {$F(I)$};
		\node [style=bn] (74) at (12.75, -2.25) {};
		\node [style=none] (76) at (12.75, -3.25) {};
		\node [style=none] (77) at (11.5, -1.25) {};
		\node [style=none] (78) at (12.75, -4) {$F(I)$};
		\node [style=none] (82) at (11.5, -1.25) {};
		\node [style=none] (83) at (11.5, 1.75) {};
		\node [style=none] (84) at (12.75, 2.75) {};
		\node [style=none] (85) at (14, 1.75) {};
		\node [style=none] (86) at (12.75, 3.75) {};
		\node [style=none] (87) at (14, -1) {};
		\node [style=none] (88) at (14, -1) {};
		\node [style=none] (89) at (10, 0) {$=$};
		\node [style=none] (92) at (7.75, 0) {};
		\node [style=none] (93) at (7.75, 4) {};
		\node [style=none] (94) at (7.75, -3) {};
		\node [style=none] (95) at (7.75, -4) {$F(I)$};
		\node [style=none] (96) at (13.5, 3) {$m$};
		\node [style=bn] (98) at (1.25, -2) {};
		\node [style=none] (99) at (1.25, -3) {};
		\node [style=none] (100) at (0, -1) {};
		\node [style=none] (101) at (1.25, -4) {$F(I)$};
		\node [style=none] (102) at (0, -1) {};
		\node [style=none] (103) at (0, 2) {};
		\node [style=none] (104) at (1.25, 3) {};
		\node [style=none] (105) at (2.5, 2) {};
		\node [style=none] (106) at (1.25, 4) {};
		\node [style=none] (107) at (2.5, -0.75) {};
		\node [style=none] (108) at (2.5, -0.75) {};
		\node [style=none] (109) at (-1.5, 0.25) {$=$};
		\node [style=none] (114) at (2, 3.25) {$m$};
		\node [style=box] (115) at (2.5, 0.5) {$\iota$};
		\node [style=box] (116) at (2.5, 0.5) {$\iota$};
	\end{pgfonlayer}
	\begin{pgfonlayer}{edgelayer}
		\draw (14.center) to (15);
		\draw (15) to (16.center);
		\draw (74) to (76.center);
		\draw [bend left] (74) to (77.center);
		\draw (83.center) to (82.center);
		\draw [bend left] (83.center) to (84.center);
		\draw [bend left] (84.center) to (85.center);
		\draw (86.center) to (84.center);
		\draw (85.center) to (87.center);
		\draw [bend right] (74) to (88.center);
		\draw (93.center) to (94.center);
		\draw (98) to (99.center);
		\draw [bend left] (98) to (100.center);
		\draw (103.center) to (102.center);
		\draw [bend left] (103.center) to (104.center);
		\draw [bend left] (104.center) to (105.center);
		\draw (106.center) to (104.center);
		\draw (105.center) to (107.center);
		\draw [bend right] (98) to (108.center);
	\end{pgfonlayer}
\end{tikzpicture}\]
	Hence, we obtain
	\[\begin{tikzpicture}
	\begin{pgfonlayer}{nodelayer}
		\node [style=bn] (0) at (4.75, -2) {};
		\node [style=bn] (1) at (6, -1) {};
		\node [style=none] (2) at (4.75, -3) {};
		\node [style=none] (3) at (3.5, -1) {};
		\node [style=none] (13) at (-9.5, 0) {$=$};
		\node [style=none] (14) at (-11, -3) {};
		\node [style=bn] (15) at (-11, -0.5) {};
		\node [style=none] (16) at (-11, 2) {};
		\node [style=none] (17) at (-12, -1.75) {$!_{F(I)}$};
		\node [style=none] (18) at (-11, -3.75) {};
		\node [style=none] (19) at (-11, -3.75) {$F(I)$};
		\node [style=none] (20) at (-11.5, 0.75) {};
		\node [style=none] (21) at (-11.5, 0.75) {$\psi_0$};
		\node [style=none] (22) at (4.75, -3.75) {$F(I)$};
		\node [style=none] (23) at (-11, 3) {$F(I)$};
		\node [style=bn] (29) at (-6.5, -2) {};
		\node [style=bn] (30) at (-7.5, -1) {};
		\node [style=none] (31) at (-6.5, -3) {};
		\node [style=none] (32) at (-6.5, -3.75) {$F(I)$};
		\node [style=none] (33) at (-8.5, 0) {};
		\node [style=none] (34) at (-6.5, 0) {};
		\node [style=none] (35) at (-7.5, 1) {};
		\node [style=none] (36) at (-7.5, 2) {};
		\node [style=none] (37) at (-5.5, -1) {};
		\node [style=box] (40) at (-5.5, 0.25) {$\iota$};
		\node [style=none] (45) at (-6.5, 3) {};
		\node [style=none] (46) at (-5.5, 2) {};
		\node [style=none] (47) at (-6.5, 4) {};
		\node [style=none] (48) at (-3.25, 0) {};
		\node [style=none] (49) at (-3.25, 2) {};
		\node [style=none] (50) at (-2, 3) {};
		\node [style=none] (52) at (-2, 4) {};
		\node [style=none] (53) at (-0.25, 1) {};
		\node [style=none] (54) at (-1.25, 0) {};
		\node [style=none] (55) at (0.75, 0) {};
		\node [style=none] (56) at (-3.25, 0) {};
		\node [style=bn] (57) at (-1.25, -2) {};
		\node [style=bn] (58) at (-2.25, -1) {};
		\node [style=none] (59) at (-1.25, -3) {};
		\node [style=none] (60) at (-1.25, -3.75) {$F(I)$};
		\node [style=none] (61) at (0.75, 0) {};
		\node [style=box] (62) at (0.75, 0) {$\iota$};
		\node [style=none] (64) at (-1.25, 0) {};
		\node [style=none] (65) at (5, 0) {};
		\node [style=none] (66) at (7, 0) {};
		\node [style=box] (67) at (7, 0) {$\iota$};
		\node [style=none] (68) at (3.5, -1) {};
		\node [style=none] (69) at (3.5, 2) {};
		\node [style=none] (70) at (4.75, 3) {};
		\node [style=none] (71) at (6, 2) {};
		\node [style=none] (72) at (4.75, 4) {};
		\node [style=none] (73) at (6, 1) {};
		\node [style=bn] (74) at (10.75, -2.25) {};
		\node [style=none] (76) at (10.75, -3.25) {};
		\node [style=none] (77) at (9.5, -1.25) {};
		\node [style=none] (78) at (10.75, -4) {$F(I)$};
		\node [style=none] (82) at (9.5, -1.25) {};
		\node [style=none] (83) at (9.5, 1.75) {};
		\node [style=none] (84) at (10.75, 2.75) {};
		\node [style=none] (85) at (12, 1.75) {};
		\node [style=none] (86) at (10.75, 3.75) {};
		\node [style=none] (88) at (12, -1.25) {};
		\node [style=none] (89) at (8, 0) {$=$};
		\node [style=none] (90) at (2, 0) {$=$};
		\node [style=none] (91) at (-4.25, 0) {$=$};
		\node [style=none] (92) at (11.5, 3) {$m$};
		\node [style=none] (93) at (5.5, 3.25) {$m$};
		\node [style=none] (94) at (-1.5, 3.25) {$m$};
		\node [style=none] (95) at (0.25, 1.25) {$m$};
		\node [style=none] (96) at (-5.5, 3.25) {$m$};
		\node [style=none] (97) at (-8, 1.25) {$m$};
		\node [style=none] (98) at (-7.75, 0.5) {};
		\node [style=none] (99) at (-6.5, 4.5) {$F(I)$};
		\node [style=none] (100) at (-2, 4.5) {$F(I)$};
		\node [style=none] (101) at (4.75, 4.5) {$F(I)$};
		\node [style=none] (102) at (10.75, 4.25) {$F(I)$};
		\node [style=none] (103) at (6, -1) {};
		\node [style=none] (104) at (12, -1.25) {};
		\node [style=bn] (105) at (12, 0.25) {};
		\node [style=none] (106) at (12, 1.75) {};
		\node [style=none] (107) at (13, -0.75) {$!_{F(I)}$};
		\node [style=none] (111) at (13, 1.25) {$\psi_0$};
	\end{pgfonlayer}
	\begin{pgfonlayer}{edgelayer}
		\draw (0) to (2.center);
		\draw [bend left=45] (0) to (3.center);
		\draw (14.center) to (15);
		\draw (15) to (16.center);
		\draw (29) to (31.center);
		\draw [bend left, looseness=0.75] (29) to (30);
		\draw [bend right=45] (33.center) to (30);
		\draw [bend left=45] (34.center) to (30);
		\draw [bend left=45] (33.center) to (35.center);
		\draw [bend left=45] (35.center) to (34.center);
		\draw (36.center) to (35.center);
		\draw [bend right] (29) to (37.center);
		\draw (46.center) to (40);
		\draw [bend left] (36.center) to (45.center);
		\draw [bend left] (45.center) to (46.center);
		\draw (47.center) to (45.center);
		\draw (49.center) to (48.center);
		\draw [bend left] (49.center) to (50.center);
		\draw (52.center) to (50.center);
		\draw [bend left=315, looseness=0.75] (53.center) to (54.center);
		\draw [bend left] (53.center) to (55.center);
		\draw (57) to (59.center);
		\draw [bend left, looseness=0.75] (57) to (58);
		\draw [bend right] (57) to (61.center);
		\draw [bend right, looseness=0.75] (56.center) to (58);
		\draw [bend left=45] (64.center) to (58);
		\draw (69.center) to (68.center);
		\draw [bend left=45, looseness=0.75] (69.center) to (70.center);
		\draw [bend left] (70.center) to (71.center);
		\draw (72.center) to (70.center);
		\draw (71.center) to (73.center);
		\draw [bend left=45] (65.center) to (73.center);
		\draw (74) to (76.center);
		\draw [bend left] (74) to (77.center);
		\draw (83.center) to (82.center);
		\draw [bend left] (83.center) to (84.center);
		\draw [bend left] (84.center) to (85.center);
		\draw (86.center) to (84.center);
		\draw [bend right] (74) to (88.center);
		\draw [bend left=45] (73.center) to (66.center);
		\draw [bend right=45] (65.center) to (103.center);
		\draw [bend left=45] (67) to (103.center);
		\draw [bend left=45] (103.center) to (0);
		\draw (40) to (37.center);
		\draw [bend left=45] (50.center) to (53.center);
		\draw (104.center) to (105);
		\draw (105) to (106.center);
	\end{pgfonlayer}
\end{tikzpicture}\]
	where the second equality follows by associativity of $m$ and the third by the associativity of $\nabla_{F(I)}$. Hence, since $\psi_0$ is the unit of the group $\langle F(I), m, \psi_0 \rangle$, and $!_{F(I)}$ is the unit of the monoid $\langle F(I), \nabla_{F(I)}, !_{F(I)} \rangle$, we obtain from the last diagram above that
	$!_{F(I)};\psi_0=\id_{F(I)}$. The assumption that $\mD$ is Markov implies that $\psi_0;!_{F(I)}=\id_I$, hence $F(I)\cong I$ and $F$ is affine. 
\end{proof}

%
%

Similarly, we can decompose Markov categories.

\begin{lemma}\label{idempmon}
Let $\mC$ be a GSM category. It is a mass category if and only if for every object $Y$, the commutative monoid 
$\langle  \mC(Y,I), \nabla_Y \comp ( - \otimes -)\comp \rho^{-1}_I, !_Y \rangle$ is idempotent. 
\end{lemma}

\begin{theorem}
Let $\mC$ be a GSM category. Then $\mC$ is a Markov category if and only if it is a mass category and a weakly Markov category.
\end{theorem}
\begin{proof}
	By Remark~\ref{rem: markov are weaklymarkov}, every Markov category is weakly Markov, and by Lemma~\ref{lemma:PR_subcategories} every Markov category is a mass category. Vice versa, if $\mC$ is a mass category and a weakly Markov category, then by Definition~\ref{defweaklymarkov} and by Lemma~\ref{idempmon}
	the commutative monoid $\langle \mC(Y,I), \nabla_Y \comp ( - \otimes -)\comp \rho^{-1}_I, !_Y \rangle$ is an idempotent group for every object $Y$. Thus $\mC(Y,I)$ is a trivial group, i.e. a singleton,  for every object $Y$, and hence $I$ is the terminal object of $\mC$.
\end{proof}

\section{A case study: semiring-weighted relations}
\label{WR}
\label{example: monad semiring}
In \cite[Section 2.3.2]{cioffogadduccitrottataxonomy} we tackled the issue of characterising instances of the semiring monad 
such that the associated Kleisli categories are either Markov or cartesian restriction categories.
We recall those results, showing how the newly introduced monads are characterised 
by the properties of the underlying semiring.

\subsection{Some facts about the semiring monad}
Consider a (commutative, distributive and unitary) semiring $(M,\oplus,\odot, 0, 1)$, i.e.  such that 
$(M,\oplus,0)$ and $(M,\odot, 1)$ are commutative monoids
and moreover  $\forall m, n, o \in M.\, m \odot (n \oplus o) = (m \odot n) \oplus (m \odot o)$. 

The well-known endofunctor
$\mathcal{M}:\mathbf{Set}\to\mathbf{Set}$ sends a set $X$ to 
 \[\mathcal{M}(X)=\left\{ h:X \to M \ |\  h\  \text{has finite support}\right\}\]
where finite support means that $h(x)\neq 0$  for a finite number of elements $x\in X$, and a function $f:X\to Y$ to the function $\tilde{f}:\mathcal{M}(X)\to \mathcal{M}(Y)$ 
mapping a $M$-valued function $h:X\to M$ with finite support 
to 
\[\tilde{f}(h)(y)= \underset{x\in f^{-1}(y)}{\bigoplus} h(x).\]

Recall that $(\mathbf{Set}, \times, \{\bullet\})$ is cartesian monoidal with respect to the cartesian product.
The above functor is lax symmetric monoidal with respect to that monoidal structure, with the  
coherence arrows 
\[\psi_{X,Y}:\mathcal{M}(X)\times \mathcal{M}(Y)\to \mathcal{M}(X\times Y)
\qquad\psi_0:\{\bullet\}\to \mathcal{M}(\{\bullet\})\]
given by $\psi_{X,Y}(h,k)(x,y)=h(x)\odot k(y)$ and $\psi_0(\bullet)(\bullet)=1$.
This functor defines a symmetric monoidal monad on $(\mathbf{Set}, \times, \{\bullet\})$,
sometimes called the \emph{semiring monad}, 
with natural transformations $\eta_X: X\to \mathcal{M}(X)$ and $\mu_X:\mathcal{M}(\mathcal{M}(X))\to\mathcal{M}(X)$ given by
\[
\eta_X(x_0)(x)=\begin{cases}
		1\ \text{if}\ x=x_0\\
		0\ \text{otherwise}
	\end{cases}
	\qquad
\mu_X(\lambda)(x)=\underset{h\in\mathcal{M}(X)}{\bigoplus}\lambda(h)\cdot h(x)
\]

\begin{remark}\label{power}
There are two other well-known endofunctors $\mathcal{M}_r$ and $\mathcal{M}_a$ on $\mathbf{Set}$ 
that are both lax symmetric monoidal with respect to the same monoidal structure $(\mathbf{Set}, \times, \{\bullet\})$. These are defined as follows
\[\mathcal{M}_r(X)=\left\{ h:X \to M \ |\  h\  \text{has support at most one and $\forall x \in X.\, h(x) = h(x) \odot h(x)$}\right\}\]
\[\mathcal{M}_a(X)=\left\{ h:X \to M \ |\  h\  \text{has finite support and $\bigoplus_{x \in X}h(x) = 1$}\right\}\]
%
Note that
$\mathcal{M}_r$ is relevant and $\mathcal{M}_a$ is affine, 
as it is easily checked out noting that $\widetilde{!_X}(h) (\bullet) =  \bigoplus_{x \in X}h(x)$ and 
$\widetilde{\nabla_X}(h) (\langle x, y \rangle) = h(x)$ if $x = y$, and $0$ otherwise.

The setting allows to recover various relation-like structures in the literature.
	If $M$ is the Boolean semiring $\{0,1\}$, then $\mathcal{M}$ is the lax symmetric monoidal 
	functor $\mathcal{P}$ associating to $X$ its finite subsets, and it is neither relevant nor affine.
	The relevant functor $\mathcal{P}_r$ is restricted to subsets of at most one element,
	while the affine functor $\mathcal{P}_a$ is restricted to subsets with at least one element. 
	If $M$ is the semiring of positive real numbers $\mathbb{R}^+$, then $\mathcal{M}$ is the monad of 
	unnormalised probability distributions, the affine functor $\mathcal{M}_a$
	is the monad of normalised probability distributions, while the relevant functor
	$\mathcal{M}_r$ associates to $X$ the functions with codomain $\{0, 1\}$, behaving
	essentially as $\mathcal{P}_r$.
\end{remark}

\subsection{Domain and mass preservation for the semiring monad}
Consider now the case in which the semiring $M$ is (multiplicative) idempotent, i.e. $\forall m \in M.\, m = m \odot m$ 
and absorptive, i.e. $\forall m \in M.\, 1 = m \oplus 1$. Then we obtain the following results for the monad $\mathcal{M}$.

\begin{proposition}
	Let $M$ be an absorptive and idempotent semiring. Then $\mathcal{M}$ is domain preserving.
\end{proposition}
\begin{proof}
In the definition of domain preservation, observe that
$\widetilde{\rho^{-1}_X}(h) (x) = h(x, \bullet)$.
Hence the canonical semigroup arrow for $\mathcal{M}(X)$ is given by the composition
$$\Delta_{\mathcal{M}(X)} = \widetilde{\rho^{-1}_X} \circ \widetilde{\id_X \otimes !_X} \circ \psi_{X, X}$$
which sends a pair of functions $h, k: X \to M$ to the function $\Delta_{\mathcal{M}(X)}(h, k): X \to M$ defined as
$\Delta_{\mathcal{M}(X)}(h, k) (x) = h(x) \odot \bigoplus_{x \in X}k(x)$.
Thus, we now have that 
$$\Delta_{\mathcal{M}(X)} \circ \nabla_{\mathcal{M}(X)} (h) (x) = h (x) \odot \bigoplus_{x \in X}h(x)$$
which is equal to $h(x)$ since $M$ absorptive and idempotent is equivalent to 
$\forall m, n \in M.\, n = n \odot (m \oplus n)$. 
\end{proof}

\begin{proposition}
	Let $M$ be an idempotent semiring. Then $\mathcal{M}$ is mass preserving.
\end{proposition}
\begin{proof}
	 As in the above computation, if $X = \{\bullet\}$, we would have 
$$\Delta_{\mathcal{M}(\{\bullet\})} \circ \nabla_{\mathcal{M}(\{\bullet\})} (h) (\bullet) = h(\bullet) \odot h(\bullet)$$
So, if $M$ is idempotent, we have $h(\bullet) \odot h(\bullet) = h(\bullet)$, which shows that $\mathcal{M}$ is unital domain preserving. Hence, since $\mathbf{Set}$ is a cartesian monoidal category,
Lemma~\ref{lemma:mass implies unital domain} implies that it is also mass preserving.
\end{proof}

\begin{example}
Absorptive and idempotent semirings are known to coincide with bounded and distributive
lattices, as obtained by the canonical order (i.e. $m \leq n$ if $m \oplus n = n$).
Consider now the Boolean semiring 
from Remark~\ref{power}. 
It is a bounded and distributive lattice, so that even though the lax symmetric monoidal  functor $\mathcal{P}$ 
associating to $X$ its finite subsets is neither relevant nor affine, it is domain preserving. 
Thus, also the associated Kleisli category of relations 
is a domain category.
Also the semiring $([0,1],\text{max},\text{min}, 0, 1)$ is a bounded and distributive lattice, hence the 
corresponding monad $\mathcal{M}$ is domain preserving,
thus the associated Kleisli category is one of Golubtsov's categories of \textit{fuzzy information transformers}.
\end{example}

\begin{example}
There are well-known examples of idempotent semirings that are 
not absorptive. 

Consider e.g. a semiring which is additive idempotent, or equivalently, such that $1\oplus1 = 1$.
Semirings whose both operators are idempotent are in fact distributive lattices.
The simplest example is the semiring with three elements
$\{0, 1, \top\}$  such that besides $1 \oplus 1 = 1$ it satisfies
$1 \oplus \top = \top = \top \odot \top$.

Consider now a semiring with characteristic 2, or equivalently, such that $1\oplus1 = 0$. A key example
is $\mathbb{Z}_2$, the semiring of integers module 2. In fact, semirings 
with characteristic 2 are rings, and if multiplicative idempotent,
they are known as Boolean rings, since
$\oplus$ can be interpreted as the XOR operator.

Also, note that these two classes are disjoint.
In fact, for a semiring being of characteristic 2 and additive idempotent
implies that it contains only one element, i.e. $0 = 1$.

\end{example}

\subsection{About weakly affine functors}
The pattern developed in the previous sections can be instantiated to obtain a large class
of weakly affine monads, extending the key example in~\cite[Example~3.7]{FritzGPT23}
\[\mathcal{M}_i(X)=\left\{ h:X \to M \ |\  h\  \text{has non-empty finite support and $\forall x \in X.\, (h(x) = 0) \vee (\exists m \in M. h(x) \odot m = 1)$}\right\}\]
For every $x$, the $m$ above is unique and we denote it by $\widehat{h(x)}$.
Also, $\mathcal{M}_i$ is lax symmetric monoidal with respect to $(\mathbf{Set}, \times, \{\bullet\})$,\
yet it is not affine, since $\mathcal{M}_i(I)$ is not isomorphic to $I = \mathcal{M}_a(I)$.

Now, for the arrow $\iota_{\mathcal{M}_i(I)}(I): \mathcal{M}_i(I) \to \mathcal{M}_i(I)$ the equations below must coincide
$$\Delta_{\mathcal{M}(I)} \circ (\id_{\mathcal{M}_i(I)} \otimes \iota_{\mathcal{M}_i(I)}) \circ \nabla_{\mathcal{M}(I)} (h) (\{\bullet\}) = 
h(\{\bullet\}) \odot \iota_{\mathcal{M}_i(I)}(h)(\{\bullet\})$$
$$\cobang_{\mathcal{M}(I)} \circ !_{\mathcal{M}(I)}(h)(\{\bullet\}) = 1$$
which is ensured by defining $\iota_{\mathcal{M}_i(I)}(h)(\{\bullet\}) = \widehat{h(\{\bullet\})}$.

\begin{remark}
The condition on being invertible holds automatically if $M$ is a semifield, i.e. if each element except $0$ has a multiplicative inverse.
In that case, $\mathcal{M}_i(X)$ includes all the functions with finite support with domain $X$, except for the empty function.
This is, e.g., the case for the semiring of positive real numbers $\mathbb{R}^+$ in Remark~\ref{power}, thus recovering the 
monad of unnormalised, non-empty probability distributions.

Also, note that if $M$ is a semifield, then $\mathcal{M}_a$ is a sub-monad of $\mathcal{M}_i$.
\end{remark}

\section{A case study: Partialisation of Markov categories}
\label{parMarkov}

In \cite{empirical,mohammed2025partializationsmarkovcategories}, the authors provide some examples of domain categories, arising as the partialisation of suitable Markov categories. {In this section, we focus on $\mathrm{Partial}(\mathrm{FinStoch})$ and provide an alternative proof of the fact that it is a domain category, using the tools developed in the previous sections.} Before proceeding, recall that $\mathrm{Partial}(\mathrm{FinStoch})$
	 is the category whose objects are finite sets and whose arrows $X\to Y$ are equivalence classes of spans
	$i: D \hookrightarrow X$, $f: D \to Y$
	where $i$ is a copyable monomorphism in $\mathrm{FinStoch}$, and $f$ is an arbitrary morphism in $\mathrm{FinStoch}$.
	
	Two arrows are equivalent if there exists an isomorphism making the diagram below on the left commute,
	and  composition is given by the equivalence class of the span $u;i,v;g$ obtained by taking the pullback below on the right
\[\begin{tikzcd}[column sep=small, row sep=small]
	& {D_f} & \\
	X && Y \\
	& {D_g}
	\arrow["i"', hook', from=1-2, to=2-1]
	\arrow["f", from=1-2, to=2-3]
	\arrow["\cong"{description}, tail reversed, from=1-2, to=3-2]
	\arrow["j", hook, from=3-2, to=2-1]
	\arrow["g"', from=3-2, to=2-3]
\end{tikzcd}
\begin{tikzcd}[column sep=small, row sep=small]
	&& E && \\
	& {D_f} && {D_g} \\
	X && Y && Z
	\arrow["u"', hook', from=1-3, to=2-2]
	\arrow["v", from=1-3, to=2-4]
	\arrow["\lrcorner"{anchor=center, pos=0.125, rotate=-45}, draw=none, from=1-3, to=3-3]
	\arrow["i"', hook', from=2-2, to=3-1]
	\arrow["f", from=2-2, to=3-3]
	\arrow["j"', hook', from=2-4, to=3-3]
	\arrow["g", from=2-4, to=3-5]
\end{tikzcd}\]
In the following, we denote by $\mathcal{D}:\mathbf{Set}\to \mathbf{Set}$ the distribution monad, i.e. 
the semiring monad $\mathcal{M}_a$ for the semiring of positive real numbers $\mathbb{R}^+$, 
with the corresponding monad structure.

\noindent
\begin{minipage}[l]{.85\linewidth}
\hspace{.4cm}
Since copyable monomorphisms in $\mathrm{FinStoch}$ are exactly arrows of the form $i;\eta:X\to \mathcal{D}(Y)$, where $i:X\to Y$ is a monomorphism in $\mathbf{Set}$ and $\eta_Y:Y\to \mathcal{D}(Y)$ is the unit of the distribution monad, which sends $y$ to $\delta_y$, then the pullback of a deterministic monomorphism along an arbitrary morphism in $\mathrm{FinStoch}$ is again a copyable monomorphism. 
For instance, in the above diagram, the pullback aside
is given by the set $E=\{x\in D_f\vert supp(f(-\vert x))\subseteq D_g\}$ and the arrow  $v:E\to D_g$ is the restriction of $f$ to $E$.
\end{minipage}
\begin{minipage}[r]{.20\linewidth}
\begin{tikzcd}
	E & {D_g} \\
	{D_f} & Y
	\arrow["v", from=1-1, to=1-2]
	\arrow["u"', hook, from=1-1, to=2-1]
	\arrow["\lrcorner"{anchor=center, pos=0.125}, draw=none, from=1-1, to=2-2]
	\arrow["j", hook, from=1-2, to=2-2]
	\arrow["f", from=2-1, to=2-2]
\end{tikzcd}
\end{minipage}

The existence of a distributive law $\gamma: \mathcal{D}(-+1)\to \mathcal{D}(-)+1$, for the distribution monad $\mathcal{D}:\mathbf{Set}\to\mathbf{Set}$ and the maybe monad $(-)+1:\mathbf{Set}\to\mathbf{Set}$ (see for instance \cite{Sokolova_2018,mio2021combining,Bonchi_2022}) ensures that $\mathcal{D}(-)+1$ is a monad  on $\mathbf{Set}$ and $\mathcal{D}$ lifts to a monad $\mathcal{D}_\mathbf{Par}$ on $\mathbf{Par}=\mathbf{Set}_{(-)+1}$, the Kleisli category of the maybe monad. Moreover, it follows by the theory of distributive laws that the two Kleisli categories are isomorphic: $\mathbf{Par}_{\mathcal{D}_\mathbf{Par}}\cong \mathbf{Set}_{\mathcal{D}+1}$.\\

 We now provide an alternative characterisation of $\mathrm{Partial}(\mathrm{FinStoch})$.
\begin{proposition}
The category $\mathrm{Partial}(\mathrm{FinStoch})$ is equivalent to the Kleisli category $\mathbf{Par}_{\mathcal{D}_\mathbf{Par}}$ of the lifted monad $\mathcal{D}_\mathbf{Par}$ on $\mathbf{Par}$.
\end{proposition}
\begin{proof}
	The objects of $\mathbf{Par}_{\mathcal{D}_\mathbf{Par}}$ are finite sets, while the arrows $X\to Y$ are given by arrows $X\to \mathcal{D}(Y)+1$ in $\mathbf{Set}$, i.e. partial functions  $X\to \mathcal{D}(Y)$. Hence, consider the functor $F:\mathbf{Par}_{\mathcal{D}_\mathbf{Par}}\to \mathrm{Partial}(\mathrm{FinStoch})$ which sends an object $X$ to itself and an arrow $f:X\to \mathcal{D}(Y)+1$ to the equivalence class of the span where $D_f=\{x\in X\vert f(x)\neq \mathrm{inr}(\bullet)\}$, $i$ is the inclusion of $D_f$ into $X$ and $f$ is the restriction of $f$ to ${D}_f$. This association is functorial. Indeed, the identity arrow $\id_X:X\to \mathcal{D}(X)+1$ sends $x$ to $\mathrm{inl}(\delta_x)$, hence $D_{\id_X}=X$ and the span associated to $\id_X$ is given by the identity span on $X$. Moreover, given two arrows $f:X\to \mathcal{D}(Y)+1$ and $g:Y\to \mathcal{D}(Z)+1$, the composition is given by the function $f;g:X\to \mathcal{D}(Z)+1$ which sends $x$ to $\mathrm{inr}(\bullet)$ if $f(x)=\mathrm{inr}(\bullet)$ or if $f(x)(y)\neq 0$ for some $y$ such that $g(y)=\mathrm{inr}(\bullet)$, or equivalently if $\mathcal{D}(g)(f(x))(\mathrm{inr}(\bullet))\neq 0$, and it sends $x\in X$ to the distribution on $Z$ given by $z\mapsto \bigoplus_{k\in \mathcal{D}(Z)}\bigoplus_{y\in \mathcal{D}(Y)} f(x)(y)\cdot k(z)$ otherwise. A simple verification shows that the composition of the spans associated to $f$ and $g$ is given by the span associated to $f;g$. Hence, $F$ is a functor. It is also an equivalence of categories, since it is full and faithful, and essentially surjective on objects.
\end{proof}

Hence, we proved that $\mathrm{Partial}(\mathrm{FinStoch})$ is equivalent to the Kleisli category ${\mathbf{Par}_{\mathcal{D}_\mathbf{Par}}}$ and also to the Kleisli category $\mathbf{Set}_{\mathcal{D}+1}$.
We now prove two properties about the monad $\mathcal{D}(-)+1:\mathbf{Set}\to \mathbf{Set}$ and $\mathcal{D}_\mathbf{Par}:\mathbf{Par}\to \mathbf{Par}$. In order to do that, we first recall the monoidal structure of these functors.\\

 $\mathcal{D}(-)+1:\mathbf{Set}\to\mathbf{Set}$ is a symmetric monoidal functor, which sends $X$ to $\mathcal{D}(X)+1$ and sends a function $f:X\to Y$ to the function $\mathcal{D}(f)+1: \mathcal{D}(X)+1\to \mathcal{D}(Y)+1$ which sends $\bullet$ to $\bullet$ and $h\in \mathcal{D}(X)$ to $\mathcal{D}(f)(h)\in \mathcal{D}(Y)$, the distribution on $Y$ given by $y\mapsto \bigoplus_{x\in f^{-1}(y)} h(x)$. 
 The monoidal structure on $\mathbf{Set}$ is given by the cartesian product, and $\overline{\psi}: (\mathcal{D}(X)+1)\times (\mathcal{D}(Y)+1)\to \mathcal{D}(X\times Y)+1$ is given by the following
\begin{itemize}
	\item $\overline{\psi}(\bullet, k)= \bullet=\overline{\psi}(h, \bullet)$;
	\item $\overline{\psi}(h, k)= \psi(h, k)$, where $\psi: \mathcal{D}(X)\times \mathcal{D}(Y)\to \mathcal{D}(X\times Y)$ is the monoidal structure of $\mathcal{D}$, which sends $(h,k)$ to the distribution on $X\times Y$ given by $(x,y)\mapsto h(x)\cdot k(y)$.
\end{itemize}
The arrow\ 
$\overline{\psi}_0: 1\to \mathcal{D}(1)+1$  sends $\bullet$ to $\delta_\bullet\in \mathcal{D}(1)$.\\

The functor $\mathcal{D}_\mathbf{Par}:\mathbf{Par}\to\mathbf{Par}$ sends $X$ to $\mathcal{D}(X)$ and sends a partial function from $X$ to $Y$, represented by an arrow  $f:X\to Y+1$, to the partial function from $\mathcal{D}(X)$ to $\mathcal{D}(Y)$, represented by the arrow $\mathcal{D}_\mathbf{Par}(f): \mathcal{D}(X)\to \mathcal{D}(Y)+1$, which is given by $\mathcal{D}(f):\mathcal{D}(X)\to \mathcal{D}(Y+1)$ post-composed with $\gamma: \mathcal{D}(Y+1)\to \mathcal{D}(Y)+1$. Explicitly, $\mathcal{D}_\mathbf{Par}(f)$ sends a distribution $h\in \mathcal{D}(X)$ to $\bullet$ if there exists $x\in f^{-1}(\mathrm{inr}(\bullet))$ such that $h(x)\neq 0$, and it sends $h$ to the distribution on $Y$ given by $y\mapsto \bigoplus_{x\in f^{-1}(\mathrm{inl}(y))} h(x)$ otherwise. 

The monoidal structure on $\mathbf{Par}$ is given by the cartesian product, and $\tilde{\psi}_{X,Y}: \mathcal{D}(X)\times \mathcal{D}(Y)\to \mathcal{D}(X\times Y) + 1$ is given by the function  $\psi_{X,Y}$ followed by the inclusion $\mathcal{D}(X\times Y)\to \mathcal{D}(X\times Y)+1$. The arrow $\tilde{\psi}_0: 1\to \mathcal{D}(1)+1$ is given by $\psi_0$ followed by the inclusion $\mathcal{D}(1)\to \mathcal{D}(1)+1$.

\begin{theorem}\label{lemma:D+1 su Set domain}
	The monad $\mathcal{D}(-)+1:\mathbf{Set}\to\mathbf{Set}$ is domain preserving.
\end{theorem}
\begin{proof}
First observe that the function $\mathcal{D}(id_X \otimes !_X) + 1: \mathcal{D}(X\times X)+1\to \mathcal{D}(X \times 1)+1$ sends $\bullet$ to $\bullet$, and $\mathcal{D}(id_X \otimes !_X)(h)+1$ is the distribution on $X\times 1$  which sends $(x,\bullet)$ to $\bigoplus_{x'\in X} h(x,x')$, for every $h\in \mathcal{D}(X\times X)$, and $x\in X$.  Finally, $\mathcal{D}(\rho_X)+1: \mathcal{D}(X\times 1)+1\to \mathcal{D}(X)+1$ sends $\bullet$ to $\bullet$, and $\mathcal{D}(\rho_X)(h)+1$ is the distribution on $X$ which sends $x$ to $h(x, \bullet)$, for every $h\in \mathcal{D}(X\times 1)$, and $x\in X$.

Now we can prove that $\mathcal{D}(-)+1$ is a domain preserving monad. Indeed, the composition
\[ \nabla_{\mathcal{D}(X)+1}; \overline{\psi}_{X,X}; (\mathcal{D}(id_X \otimes !_X)+1); (\mathcal{D}(\rho_X)+1) \]
sends $\bullet$ to $\bullet$, and sends a distribution $h\in \mathcal{D}(X)$ to the distribution which maps $x$ to $h(x)\cdot\bigoplus_{x'\in X}h(x')$. Since $h\in \mathcal{D}(X)$, the latter product is equal to $h(x)$. Hence, the composition above is the identity on $\mathcal{D}(X)+1$, and thus $\mathcal{D}(-)+1$ is a domain preserving monad.
\end{proof}

\begin{proposition}\label{lemma:D+1 non affine}
	The monad $\mathcal{D}(-)+1:\mathbf{Set}\to\mathbf{Set}$  is not affine.
\end{proposition}
\noindent
\begin{minipage}[l]{.60\linewidth}
\begin{proof}
To see that $\mathcal{D}(-)+1$ is not affine, observe that in the diagram aside
$1$ is sent to $1$ by the top horizontal arrow, while it is sent to $\delta_\bullet\in \mathcal{D}(1)$ 
by the composition of the other two arrows. Hence, the diagram aside does not commute and 
$\mathcal{D}(-)+1$ is not affine.
\end{proof}
\end{minipage}
\hspace{.4cm}
\begin{minipage}[r]{.35\linewidth}
\begin{tikzcd}[column sep=tiny]
	\mathcal{D}(X)+1 && {\mathcal{D}(1)+1} \\
& 1
\arrow["{\mathcal{D}(!_X)+1}", from=1-1, to=1-3]
\arrow["{!_{\mathcal{D}(X)+1}}"', from=1-1, to=2-2]
\arrow["{\tilde{\psi}_0}"', from=2-2, to=1-3]
\end{tikzcd}
\end{minipage}

\begin{theorem}\label{lemma:D su par affine}
The monad $\mathcal{D}_\mathbf{Par}:\mathbf{Par}\to\mathbf{Par}$ is affine.
\end{theorem}
\noindent
\begin{minipage}[l]{.60\linewidth}
\begin{proof}
Recall that $!_X:X\to 1$ in $\mathbf{Par}$ is the function sending any element of $X$ to the element $\bullet$ of $1$.
Now consider the diagram aside in $\mathbf{Par}$. Since $\mathcal{D}_\mathbf{Par}(!_X): \mathcal{D}_\mathbf{Par}(X)\to \mathcal{D}_\mathbf{Par}(1)+1$ sends a distribution $h\in \mathcal{D}_\mathbf{Par}(X)$ to $\delta_\bullet\in \mathcal{D}_\mathbf{Par}(1)$, the diagram aside commutes and hence $\mathcal{D}_\mathbf{Par}$ is affine.
\end{proof}
\end{minipage}
\hspace{.4cm}
\begin{minipage}[r]{.35\linewidth}
\begin{tikzcd}
	\mathcal{D}_\mathbf{Par}(X) && {\mathcal{D}_\mathbf{Par}(1)} \\
& 1
\arrow["{\mathcal{D}_\mathbf{Par}(!_X)}", from=1-1, to=1-3]
\arrow["{!_{\mathcal{D}_\mathbf{Par}(X)}}"', from=1-1, to=2-2]
\arrow["{\tilde{\psi}_0}"', from=2-2, to=1-3]
\end{tikzcd}
\end{minipage}

\vspace{.1cm}
Observe that the monad $\mathcal{D}_\mathbf{Par}$ is an example of an affine monad in the sense of \cite{cioffogadduccitrottataxonomy} on a restriction category. This notion generalises the notion of affine monad for cartesian categories introduced in \cite{Jacobs1994}.

\begin{corollary} 
	The category $\mathrm{Partial}(\mathrm{FinStoch})$ is a domain category.
\end{corollary}
\begin{proof}
	It follows by Lemma \ref{lemma:D su par affine} and Theorem \ref{thm: restriction + special implies kl in PR}, by observing that every affine monad is domain preserving by Corollary \ref{cor:thm}. It also follows by Lemma \ref{lemma:D+1 su Set domain} and Theorem \ref{thm: restriction + special implies kl in PR}.
\end{proof}

\section{Enrichment and oplax cartesianity}
\label{oplax}
Now we focus on the canonical poset-enrichment of domain categories~\cite{lorenz2023causalmodelsstringdiagrams}, and on how it interacts with the monoidal structure,  see also \cite{empirical,Di_Lavore_2025order}.

%


\begin{definition}
	Let $\mC$ be a GSM category and $f,g:X\to Y$. We say that $g$ \emph{extends} $f$, in symbol $g \sqsupseteq f$, if
	\begin{center}
		\begin{tikzpicture}
	\begin{pgfonlayer}{nodelayer}
		\node [style=none] (0) at (-3, 0) {};
		\node [style=none] (1) at (-3, -1) {};
		\node [style=box] (2) at (-2, 1.5) {$g$};
		\node [style=box] (3) at (-4, 1.5) {$f$};
		\node [style=none] (4) at (-2, 1) {};
		\node [style=none] (5) at (-4, 1) {};
		\node [style=none] (6) at (-4, 3) {};
		\node [style=none] (7) at (-2, 3) {};
		\node [style=none] (8) at (0, 0) {$=$};
		\node [style=none] (9) at (2, -1) {};
		\node [style=none] (10) at (2, 3) {};
		\node [style=box] (11) at (2, 1) {$f$};
		\node [style=none] (12) at (2.5, -1) {$X$};
		\node [style=none] (13) at (2.5, 3) {$Y$};
		\node [style=none] (14) at (-2.5, -1) {$X$};
		\node [style=none] (15) at (-1.5, 3) {$Y$};
		\node [style=bn] (16) at (-3, 0) {};
		\node [style=bn] (17) at (-4, 3) {};
	\end{pgfonlayer}
	\begin{pgfonlayer}{edgelayer}
		\draw (10.center) to (9.center);
		\draw (0.center) to (1.center);
		\draw [bend left=45] (0.center) to (5.center);
		\draw [bend right=45] (0.center) to (4.center);
		\draw (5.center) to (6.center);
		\draw (4.center) to (7.center);
	\end{pgfonlayer}
\end{tikzpicture}
	\end{center}
\end{definition}

We now recall a result that is basically in~\cite[Lemma 98]{lorenz2023causalmodelsstringdiagrams}.
%

\begin{proposition}\label{prop: poset homset}
Let $\mC$ be a domain category. Then $\sqsupseteq$ is a partial order on the hom-sets of $\mC$. 
\end{proposition}


The enrichment is not preserved by the monoidal structure, yet
the result below holds~\cite[Prop. 2.16]{empirical}.

\begin{proposition}\label{prop:poset enrichment}
	Let $\mC$ be a positive domain category. Then it is a poset enriched monoidal category with respect to the partial order $\sqsupseteq$.
	\end{proposition}

\begin{definition}\label{def oplax cartesian cat}
A \emph{posetal oplax cartesian category}  is a poset-enriched GSM category $\mC$ such that every arrow is \emph{oplax copyable} and \emph{oplax discardable}, 
i.e. the following inequalities hold for every arrow $\freccia{X}{f}{Y}$
	\ctikzfig{oplax_cart_def}
\end{definition}
\begin{lemma}\label{lemma:oplax implies domain}
	Let $\mC$ be a posetal oplax cartesian category. Then it is a domain category.
\end{lemma}
\begin{proof}
	Consider the following derivation, where the first inequality follows from the oplax copyability of $f$, and the second one follows from the oplax discardability of $f$
 	\begin{center}
		\begin{tikzpicture}
	\begin{pgfonlayer}{nodelayer}
		\node [style=none] (0) at (0.75, -3) {$X$};
		\node [style=none] (1) at (1.25, -1.75) {};
		\node [style=bn] (2) at (1.25, 0) {};
		\node [style=none] (3) at (0.25, 1) {};
		\node [style=none] (4) at (2.25, 1) {};
		\node [style=box] (13) at (1.25, -1.5) {$f$};
		\node [style=none] (14) at (1.25, -3) {};
		\node [style=none] (15) at (0.75, -0.5) {$Y$};
		\node [style=none] (16) at (3.25, -1) {$\leq$};
		\node [style=bn] (17) at (6, -1.5) {};
		\node [style=none] (18) at (5, -0.5) {};
		\node [style=none] (19) at (7, -0.5) {};
		\node [style=none] (20) at (5.5, -3) {$X$};
		\node [style=none] (21) at (6, -3) {};
		\node [style=box] (22) at (5, -0.5) {$f$};
		\node [style=box] (23) at (7, -0.5) {$f$};
		\node [style=none] (24) at (5, 1) {};
		\node [style=none] (25) at (7, 1) {};
		\node [style=none] (26) at (4.5, 1) {$Y$};
		\node [style=none] (33) at (-3, -3) {$X$};
		\node [style=none] (34) at (-2.5, -1.25) {};
		\node [style=box] (36) at (-2.5, -1) {$f$};
		\node [style=none] (37) at (-2.5, -3) {};
		\node [style=none] (38) at (-3, 1) {$Y$};
		\node [style=none] (39) at (-1, -1.25) {$=$};
		\node [style=bn] (40) at (2.25, 1) {};
		\node [style=none] (41) at (-2.5, 1) {};
		\node [style=bn] (42) at (7, 1) {};
		\node [style=none] (43) at (10.5, -3) {$X$};
		\node [style=none] (44) at (11, -3) {};
		\node [style=bn] (45) at (11, -1.75) {};
		\node [style=none] (46) at (10, -0.5) {};
		\node [style=none] (47) at (12, -0.5) {};
		\node [style=box] (48) at (10, -0.5) {$f$};
		\node [style=none] (49) at (11, -3) {};
		\node [style=none] (50) at (9.5, 1) {$Y$};
		\node [style=bn] (51) at (12, 1) {};
		\node [style=none] (52) at (10, 1) {};
		\node [style=none] (53) at (12, 1) {};
		\node [style=none] (54) at (8.75, -1) {$\leq$};
		\node [style=none] (55) at (14.5, -3) {$X$};
		\node [style=none] (56) at (15, -1.25) {};
		\node [style=box] (57) at (15, -1) {$f$};
		\node [style=none] (58) at (15, -3) {};
		\node [style=none] (59) at (14.5, 1) {$Y$};
		\node [style=none] (60) at (13.5, -1) {$=$};
		\node [style=none] (61) at (15, 1) {};
	\end{pgfonlayer}
	\begin{pgfonlayer}{edgelayer}
		\draw (2) to (1.center);
		\draw [in=165, out=-90] (3.center) to (2);
		\draw [in=270, out=15] (2) to (4.center);
		\draw (13) to (14.center);
		\draw [in=165, out=-90] (18.center) to (17);
		\draw [in=270, out=15] (17) to (19.center);
		\draw (21.center) to (17);
		\draw (24.center) to (22);
		\draw (25.center) to (23);
		\draw (36) to (37.center);
		\draw (41.center) to (36);
		\draw (45) to (44.center);
		\draw [in=165, out=-90] (46.center) to (45);
		\draw [in=270, out=15] (45) to (47.center);
		\draw (52.center) to (48);
		\draw (53.center) to (47.center);
		\draw (57) to (58.center);
		\draw (61.center) to (57);
	\end{pgfonlayer}
\end{tikzpicture}
	\end{center}
Hence, the poset-enrichment implies the domain equation.
\end{proof}

We now explore how the canonical order interacts with oplax cartesianity.

\begin{proposition}
Let $\mC$ be a domain category. Then with respect to the partial order $\sqsupseteq$
\begin{itemize}
	\item every arrow is oplax discardable;
	\item if an arrow is oplax copyable then it is copyable.
	 \end{itemize}
\end{proposition}
\begin{proof}
	The first property is an immediate consequence of the definition of $\sqsupseteq$.
	
	As for the second, recall for an arrow $f:X\to Y$ to be oplax copyable with respect to the partial order $\sqsupseteq$
	means that the equivalence below holds
	\begin{center}
		\scalebox{0.8}{\begin{tikzpicture}
	\begin{pgfonlayer}{nodelayer}
		\node [style=none] (0) at (-0.5, -3) {$X$};
		\node [style=none] (1) at (0, -3) {};
		\node [style=bn] (2) at (0, -0.25) {};
		\node [style=none] (3) at (-1, 1) {};
		\node [style=none] (4) at (1, 1) {};
		\node [style=box] (13) at (0, -1.5) {$f$};
		\node [style=none] (14) at (0, -3) {};
		\node [style=none] (15) at (-1.5, 1) {$Y$};
		\node [style=none] (16) at (2.75, -1) {$\sqsubseteq$};
		\node [style=bn] (17) at (6, -1.5) {};
		\node [style=none] (18) at (5, -0.5) {};
		\node [style=none] (19) at (7, -0.5) {};
		\node [style=none] (20) at (5.5, -3) {$X$};
		\node [style=none] (21) at (6, -3) {};
		\node [style=box] (22) at (5, -0.5) {$f$};
		\node [style=box] (23) at (7, -0.5) {$f$};
		\node [style=none] (24) at (5, 1) {};
		\node [style=none] (25) at (7, 1) {};
		\node [style=none] (26) at (4.5, 1) {$Y$};
		\node [style=none] (27) at (7.5, 1) {$Y$};
		\node [style=bn] (40) at (11.5, 3) {};
		\node [style=none] (43) at (-1, 1) {};
		\node [style=none] (44) at (1, 1) {};
		\node [style=none] (45) at (9.75, -1) {iff};
		\node [style=bn] (46) at (16.5, 0.5) {};
		\node [style=none] (47) at (15.5, 1.5) {};
		\node [style=none] (48) at (17.5, 1.5) {};
		\node [style=none] (49) at (17, -0.5) {$X$};
		\node [style=none] (50) at (16.5, -0.5) {};
		\node [style=box] (51) at (15.5, 1.5) {$f$};
		\node [style=box] (52) at (17.5, 1.5) {$f$};
		\node [style=none] (53) at (15.5, 3) {};
		\node [style=none] (54) at (17.5, 3) {};
		\node [style=none] (55) at (15, 3) {$Y$};
		\node [style=none] (56) at (18, 3) {$Y$};
		\node [style=bn] (57) at (14.5, -2.75) {};
		\node [style=none] (58) at (12.5, -1) {};
		\node [style=none] (59) at (16.5, -0.5) {};
		\node [style=none] (60) at (14, -4.25) {$X$};
		\node [style=none] (61) at (14.5, -4.25) {};
		\node [style=bn] (62) at (12.5, 1.75) {};
		\node [style=none] (63) at (11.5, 3) {};
		\node [style=none] (64) at (13.5, 3) {};
		\node [style=none] (65) at (12, -1) {$X$};
		\node [style=none] (66) at (12.5, -1) {};
		\node [style=box] (68) at (12.5, 0.5) {$f$};
		\node [style=none] (69) at (11.5, 3) {};
		\node [style=none] (70) at (13.5, 3) {};
		\node [style=bn] (73) at (13.5, 3) {};
		\node [style=none] (74) at (18.75, -0.5) {=};
		\node [style=none] (75) at (20.75, -3) {$X$};
		\node [style=none] (76) at (21.25, -3) {};
		\node [style=bn] (77) at (21.25, -0.25) {};
		\node [style=none] (78) at (20.25, 1) {};
		\node [style=none] (79) at (22.25, 1) {};
		\node [style=box] (80) at (21.25, -1.5) {$f$};
		\node [style=none] (81) at (21.25, -3) {};
		\node [style=none] (82) at (19.75, 1) {$Y$};
		\node [style=none] (83) at (20.25, 1) {};
		\node [style=none] (84) at (22.25, 1) {};
		\node [style=none] (85) at (22.75, 1) {$Y$};
	\end{pgfonlayer}
	\begin{pgfonlayer}{edgelayer}
		\draw (2) to (1.center);
		\draw [in=165, out=-90] (3.center) to (2);
		\draw [in=270, out=15] (2) to (4.center);
		\draw [in=165, out=-90] (18.center) to (17);
		\draw [in=270, out=15] (17) to (19.center);
		\draw (21.center) to (17);
		\draw (24.center) to (22);
		\draw (25.center) to (23);
		\draw [in=165, out=-90] (47.center) to (46);
		\draw [in=270, out=15] (46) to (48.center);
		\draw (50.center) to (46);
		\draw (53.center) to (51);
		\draw (54.center) to (52);
		\draw [in=180, out=-90, looseness=0.75] (58.center) to (57);
		\draw [in=270, out=15] (57) to (59.center);
		\draw (61.center) to (57);
		\draw [in=165, out=-90] (63.center) to (62);
		\draw [in=270, out=15] (62) to (64.center);
		\draw (66.center) to (62);
		\draw (77) to (76.center);
		\draw [in=165, out=-90] (78.center) to (77);
		\draw [in=270, out=15] (77) to (79.center);
	\end{pgfonlayer}
\end{tikzpicture}}
	\end{center}
	However, by associativity and the domain equation we get
	\begin{center}
\scalebox{0.8}{\begin{tikzpicture}
	\begin{pgfonlayer}{nodelayer}
		\node [style=none] (74) at (19.25, -2) {=};
		\node [style=bn] (75) at (37, -1.75) {};
		\node [style=none] (76) at (36, -0.75) {};
		\node [style=none] (77) at (38, -0.75) {};
		\node [style=none] (78) at (36.5, -3.25) {$X$};
		\node [style=none] (79) at (37, -3.25) {};
		\node [style=box] (80) at (36, -0.75) {$f$};
		\node [style=box] (81) at (38, -0.75) {$f$};
		\node [style=none] (82) at (36, 0.75) {};
		\node [style=none] (83) at (38, 0.75) {};
		\node [style=none] (84) at (35.5, 0.75) {$Y$};
		\node [style=none] (85) at (38.5, 0.75) {$Y$};
		\node [style=bn] (86) at (12.25, 2) {};
		\node [style=bn] (87) at (17.25, -0.5) {};
		\node [style=none] (88) at (16.25, 0.5) {};
		\node [style=none] (89) at (18.25, 0.5) {};
		\node [style=none] (91) at (17.25, -1.5) {};
		\node [style=box] (92) at (16.25, 0.5) {$f$};
		\node [style=box] (93) at (18.25, 0.5) {$f$};
		\node [style=none] (94) at (16.25, 2) {};
		\node [style=none] (95) at (18.25, 2) {};
		\node [style=none] (96) at (15.75, 2) {$Y$};
		\node [style=none] (97) at (18.75, 2) {$Y$};
		\node [style=bn] (98) at (15.25, -3.75) {};
		\node [style=none] (99) at (13.25, -1.5) {};
		\node [style=none] (100) at (17.25, -1.5) {};
		\node [style=none] (101) at (14.75, -5.25) {$X$};
		\node [style=none] (102) at (15.25, -5.25) {};
		\node [style=bn] (103) at (13.25, 0.75) {};
		\node [style=none] (104) at (12.25, 2) {};
		\node [style=none] (105) at (14.25, 2) {};
		\node [style=none] (107) at (13.25, -1.5) {};
		\node [style=box] (108) at (13.25, -0.5) {$f$};
		\node [style=none] (109) at (12.25, 2) {};
		\node [style=none] (110) at (14.25, 2) {};
		\node [style=bn] (113) at (14.25, 2) {};
		\node [style=bn] (115) at (24.75, -0.5) {};
		\node [style=none] (116) at (23.75, 0.5) {};
		\node [style=none] (117) at (25.75, 0.5) {};
		\node [style=none] (119) at (24.75, -1.5) {};
		\node [style=box] (120) at (23.75, 0.5) {$f$};
		\node [style=box] (121) at (25.75, 0.5) {$f$};
		\node [style=none] (122) at (23.75, 2) {};
		\node [style=none] (123) at (25.75, 2) {};
		\node [style=none] (124) at (23.25, 2) {$Y$};
		\node [style=none] (125) at (26.25, 2) {$Y$};
		\node [style=bn] (126) at (22.75, -3.75) {};
		\node [style=none] (127) at (20.75, -1.5) {};
		\node [style=none] (128) at (24.75, -1.5) {};
		\node [style=none] (129) at (22.25, -5.25) {$X$};
		\node [style=none] (130) at (22.75, -5.25) {};
		\node [style=bn] (131) at (20.75, 2) {};
		\node [style=none] (135) at (20.75, -1.5) {};
		\node [style=box] (136) at (20.75, 0) {$f$};
		\node [style=bn] (141) at (29.5, -0.5) {};
		\node [style=none] (142) at (28.5, 0.5) {};
		\node [style=none] (143) at (30.5, 0.5) {};
		\node [style=none] (145) at (29.5, -1.5) {};
		\node [style=box] (146) at (28.5, 0.5) {$f$};
		\node [style=box] (147) at (30.5, 0.5) {$f$};
		\node [style=none] (148) at (28.5, 2) {};
		\node [style=none] (149) at (30.5, 2) {};
		\node [style=none] (151) at (31, 2) {$Y$};
		\node [style=bn] (152) at (31.5, -3.75) {};
		\node [style=none] (153) at (29.5, -1.5) {};
		\node [style=none] (154) at (33.25, -1.25) {};
		\node [style=none] (155) at (31, -5.25) {$X$};
		\node [style=none] (156) at (31.5, -5.25) {};
		\node [style=box] (160) at (33.25, 0) {$f$};
		\node [style=bn] (162) at (28.5, 2) {};
		\node [style=none] (163) at (33.25, 2) {};
		\node [style=none] (164) at (27.25, -2) {=};
		\node [style=none] (165) at (34.75, -2) {=};
		\node [style=none] (166) at (33.75, 2) {$Y$};
	\end{pgfonlayer}
	\begin{pgfonlayer}{edgelayer}
		\draw [in=165, out=-90] (76.center) to (75);
		\draw [in=270, out=15] (75) to (77.center);
		\draw (79.center) to (75);
		\draw (82.center) to (80);
		\draw (83.center) to (81);
		\draw [in=165, out=-90] (88.center) to (87);
		\draw [in=270, out=15] (87) to (89.center);
		\draw (91.center) to (87);
		\draw (94.center) to (92);
		\draw (95.center) to (93);
		\draw [in=165, out=-90] (99.center) to (98);
		\draw [in=270, out=15] (98) to (100.center);
		\draw (102.center) to (98);
		\draw [in=165, out=-90] (104.center) to (103);
		\draw [in=270, out=15] (103) to (105.center);
		\draw (107.center) to (103);
		\draw [in=165, out=-90] (116.center) to (115);
		\draw [in=270, out=15] (115) to (117.center);
		\draw (119.center) to (115);
		\draw (122.center) to (120);
		\draw (123.center) to (121);
		\draw [in=165, out=-90] (127.center) to (126);
		\draw [in=270, out=15] (126) to (128.center);
		\draw (130.center) to (126);
		\draw (135.center) to (131);
		\draw [in=165, out=-90] (142.center) to (141);
		\draw [in=270, out=15] (141) to (143.center);
		\draw (145.center) to (141);
		\draw (148.center) to (146);
		\draw (149.center) to (147);
		\draw [in=165, out=-90] (153.center) to (152);
		\draw [in=270, out=15] (152) to (154.center);
		\draw (156.center) to (152);
		\draw (163.center) to (154.center);
	\end{pgfonlayer}
\end{tikzpicture}}
\end{center}
	hence $f$ is actually copyable.
	\end{proof}
	
The next result states that if the partial order $\sqsupseteq$ gives rise to a posetal oplax category, then it collapses.

\begin{corollary}
Let $\mC$ be a domain category. If it is a posetal oplax category with respect to the partial order $\sqsupseteq$,
then it is a cartesian restriction category.
\end{corollary}

Thus, it mimics the natural order on partial functions, as further strengthened by next corollary.

	\begin{corollary}
		Let $\mC$ be a Markov category. If it is a posetal oplax category with respect to the partial order $\sqsupseteq$, then it is a cartesian monoidal category.
	\end{corollary}

	The simple observation below shows that the relation $\sqsupseteq$ is contained in any other poset-enrichment and is compatible with the monoidal structure in the sense of Definition~\ref{def oplax cartesian cat}.

	\begin{proposition}
		Let $\mC$ be an oplax cartesian category with respect to a partial order $\leq$. Then for every pair of parallel arrows $f,g:X\to Y$, if $f\sqsubseteq g$ then $f\leq g$.
	\end{proposition}
	\begin{proof}
		By Lemma \ref{lemma:oplax implies domain} we have that $\mC$ is a domain category, hence Proposition \ref{prop: poset homset} implies that $\sqsupseteq$ defines a partial order on $\mC(X,Y)$. The following derivation implies the statement
		\begin{center}
			\begin{tikzpicture}
	\begin{pgfonlayer}{nodelayer}
		\node [style=none] (0) at (-3, 0) {};
		\node [style=none] (1) at (-3, -1) {};
		\node [style=box] (2) at (-2, 1.5) {$g$};
		\node [style=box] (3) at (-4, 1.5) {$f$};
		\node [style=none] (4) at (-2, 1) {};
		\node [style=none] (5) at (-4, 1) {};
		\node [style=none] (6) at (-4, 3) {};
		\node [style=none] (7) at (-2, 3) {};
		\node [style=none] (8) at (-5.75, 1) {$=$};
		\node [style=none] (9) at (-7.25, -1) {};
		\node [style=none] (10) at (-7.25, 3) {};
		\node [style=box] (11) at (-7.25, 1) {$f$};
		\node [style=none] (12) at (-6.75, -1) {$X$};
		\node [style=none] (13) at (-6.75, 3) {$Y$};
		\node [style=none] (14) at (-2.5, -1) {$X$};
		\node [style=none] (15) at (-1.5, 3) {$Y$};
		\node [style=bn] (16) at (-3, 0) {};
		\node [style=bn] (17) at (-4, 3) {};
		\node [style=none] (18) at (-0.5, 1) {$\le$};
		\node [style=none] (19) at (2.25, 0) {};
		\node [style=none] (20) at (2.25, -1) {};
		\node [style=box] (21) at (3.25, 1.5) {$g$};
		\node [style=none] (23) at (3.25, 1) {};
		\node [style=none] (24) at (1.25, 1) {};
		\node [style=none] (25) at (1.25, 3) {};
		\node [style=none] (26) at (3.25, 3) {};
		\node [style=none] (27) at (2.75, -1) {$X$};
		\node [style=none] (28) at (3.75, 3) {$Y$};
		\node [style=bn] (29) at (2.25, 0) {};
		\node [style=bn] (30) at (1.25, 3) {};
		\node [style=none] (31) at (4.75, 1) {$=$};
		\node [style=none] (32) at (6, -1) {};
		\node [style=none] (33) at (6, 3) {};
		\node [style=box] (34) at (6, 1) {$g$};
		\node [style=none] (35) at (6.5, -1) {$X$};
		\node [style=none] (36) at (6.5, 3) {$Y$};
		\node [style=none] (37) at (-11, 1) {$f\sqsubseteq g$};
		\node [style=none] (38) at (-9, 1) {$\iff$};
	\end{pgfonlayer}
	\begin{pgfonlayer}{edgelayer}
		\draw (10.center) to (9.center);
		\draw (0.center) to (1.center);
		\draw [bend left=45] (0.center) to (5.center);
		\draw [bend right=45] (0.center) to (4.center);
		\draw (5.center) to (6.center);
		\draw (4.center) to (7.center);
		\draw (19.center) to (20.center);
		\draw [bend left=45] (19.center) to (24.center);
		\draw [bend right=45] (19.center) to (23.center);
		\draw (24.center) to (25.center);
		\draw (23.center) to (26.center);
		\draw (33.center) to (32.center);
	\end{pgfonlayer}
\end{tikzpicture} 
		\end{center}
		where the middle inequality follows from oplax discardability.
		\end{proof}
	 Thus, for positive oplax cartesian categories the order $\sqsubseteq$ is the minimal poset-enrichment.

%

\section{Conclusions and further works}

The present work builds on the results presented in~\cite{cioffogadduccitrottataxonomy}, 
where it was proposed a taxonomy for some variants of symmetric monoidal categories, 
such as Markov and restriction categories, which in recent years had been introduced 
with computational and graphical aims.
The survey was further complemented by a series of results concerning the 
structures of Kleisli categories, putting some order also on those
variants of (order-enriched) symmetric monoidal monads proposed in the literature.
This paper extends the taxonomy, including the already 
known weakly Markov categories and the new domain and 
mass categories, and establishing a connection with the corresponding
monads and their Kleisli categories.

Domain and mass categories are instances of gs-monoidal categories, hence part of the taxonomy,
yet they are intermediate between Markov and restriction categories, making more precise
the underlying algebraic structure in terms of special and Hopf monoid objects.
Indeed, an outcome of such description is the characterisation of Markov categories via weakly 
Markov and mass categories.

Our work thus extends \cite{cioffogadduccitrottataxonomy},  
investigating as promised there ``alternative notions of Markov categories~\cite{LavoreR23,FritzGPT23}  
and affine monads~\cite{Jacobs16,FritzGPT23}, aimed 
at distilling a categorical presentation of probability theory, see e.g.~\cite{Jacobs18}.''
We foresee a few research threads we plan to explore. On the restriction categories side, 
the connections between the axiom (R.4) and positivity, whose relationship is hinted at in 
Proposition~\ref{positive as restriction}, and between our monads and the classifying monads
of~\cite{Cockett03}.
On the categorical probability side, we plan to investigate if
partial Markov categories~\cite{LavoreR23,abs-2502-03477}
fit into our taxonomy, as well as enlarging it 
towards traced gs-monoidal categories~\cite{Joyal_tracedcategories,CorradiniGadducci99b}, 
and to explore completeness theorems for functorial semantics, see~\cite{FritzGCT23}.


\bibliography{references}
\bibliographystyle{./entics}

\appendix

 \section{Lax monoidal functors and commutative monads}
\label{sec:lax_app}

This section recalls the definitions of lax monoidal functor~\cite{aguiar2010} and of symmetric monoidal monad ~\cite{Kock70,Kock72}.
Throughout, $\mC$ and $\mD$ are symmetric monoidal categories with tensor functor $\otimes$ and monoidal unit $I$, and we assume that 
$\otimes$ strictly associates without loss of generality to keep the diagrams simple.
Left and right unitors are denoted by $\lambda$ and $\rho$, respectively, and braidings by $\gamma$.

\begin{definition}\label{def:lax monoidal functor}
Let $\mC$ and ${\mD}$ be monoidal categories. 
A functor $\freccia{\mC}{F}{\mD}$ is \emph{lax monoidal} if it is equipped with a natural transformation 
\[\freccia{\otimes \circ \, (F\times F)}{\psi}{F\circ \otimes}\]
and an arrow $\freccia{I}{\psi_0}{F(I)}$ such that the following associativity and unitality diagrams commute
\[\begin{tikzcd}[column sep=2ex, row sep=2ex]
	{F(A)\otimes F(B)\otimes F(C)} &&& {F(A)\otimes F(B\otimes C)} \\
	\\
	{F(A\otimes B)\otimes F(C)} &&& {F(A\otimes B\otimes C)}
	\arrow["{\id\otimes \,\psi_{B,C}}", from=1-1, to=1-4]
	\arrow["{\psi_{A,B}\otimes {}\id}"', from=1-1, to=3-1]
	\arrow["{\psi_{A\otimes B,C}}"', from=3-1, to=3-4]
	\arrow["{\psi_{A,B\otimes C}}", from=1-4, to=3-4]
\end{tikzcd}
\hspace{.2cm}
\begin{tikzcd}[column sep=0.8ex, row sep=2ex]
	{I\otimes F(A)} && F(A) && {F(A)\otimes  I} && F(A) \\
	\\
	{F(I)\otimes F(A)} && {F(I\otimes A)} && {F(A)\otimes F(I)} && {F(A\otimes  I)}
	\arrow["{\psi_{I,A}}"', from=3-1, to=3-3]
	\arrow["{F(\lambda_A)}", from=1-3, to=3-3]
	\arrow["{\psi_0\otimes {}\id}"', from=1-1, to=3-1]
	\arrow["{\lambda_{FA}}"', from=1-3, to=1-1]
	\arrow["{\rho_{FA}}"', from=1-7, to=1-5]
	\arrow["{\id\otimes \psi_0}"', from=1-5, to=3-5]
	\arrow["{\psi_{A,I}}"', from=3-5, to=3-7]
	\arrow["{F(\rho_A)}", from=1-7, to=3-7]
\end{tikzcd}\]


If $\mC$ and ${\mD}$ are symmetric monoidal categories, then
$F$ is a \emph{lax symmetric monoidal} functor if also the following diagram commutes

\begin{center}
\begin{tikzcd}
	{F(A)\otimes F(B)} && {F(B)\otimes F(A)} \\
	{F(A\otimes B)} && {F(B\otimes A)}
	\arrow["{\gamma^{\mathcal{D}}_{FA,FB}}", from=1-1, to=1-3]
	\arrow["{\psi_{A,B}}", from=1-1, to=2-1]
	\arrow["{\psi_ {B,A}}"', from=1-3, to=2-3]
	\arrow["{F(\gamma^{\mathcal{C}}_{A,B})}"', from=2-1, to=2-3]
\end{tikzcd}
\end{center}
\end{definition}

For example, if $\mC$ is the terminal monoidal category with only one object $I$, then $F$ is a monoid in $\mD$.


\begin{definition}
	A \emph{monoidal transformation} between lax monoidal functors $\freccia{(F,\psi_0,\psi)}{\epsilon}{(F',\psi'_0,\psi')}:\mC\to\mD$ is a family of arrows $\epsilon_X:F(X)\to F'(X)$, for $X\in\mC$, satisfying
\[
\begin{tikzcd}[column sep=tiny]
	{F(X)\otimes F(Y)} && {F'(X)\otimes F'(Y)} && I && {F(I)} \\
	{F(X\otimes Y)} && {F'(X\otimes Y)} &&& {F'(I)}
	\arrow["{\epsilon_X\otimes\epsilon_Y}", from=1-1, to=1-3]
	\arrow["\psi"', from=1-1, to=2-1]
	\arrow["{\psi'}", from=1-3, to=2-3]
	\arrow["{\psi_0}", from=1-5, to=1-7]
	\arrow["{\psi'_0}"', from=1-5, to=2-6]
	\arrow["{\epsilon_I}", from=1-7, to=2-6]
	\arrow["{\epsilon_{X\otimes Y}}"', from=2-1, to=2-3]
\end{tikzcd}\]
If $\epsilon$ is a natural transformation between the functors $F$ and $F'$, it is called \emph{monoidal natural transformation}.
\end{definition}

%

\begin{definition}
\label{dfn: symmetric monoidal monad}
	Let $\mathcal{C}$ be a symmetric monoidal category and $T : \mathcal{C} \to \mathcal{C}$ be a monad carrying the structure of a lax symmetric monoidal functor with structure maps $\freccia{\otimes \circ \, (T\times T)}{c}{T\circ \otimes}$ and $u : I \to TI$. Then $T$ is a \emph{symmetric monoidal monad} if $u = \eta_I$ and the following two diagrams commute
\[\begin{tikzcd}
& X \otimes Y \arrow[dl, "\eta \otimes \eta"'] \arrow[dr, "\eta"] & \\
TX \otimes TY \arrow[rr, "c"'] & & T(X \otimes Y)
\end{tikzcd}
\quad 
\begin{tikzcd}
TTX \otimes TTY \arrow[r, "c"] \arrow[d, "\mu \otimes \mu"'] & T(TX \otimes TY) \arrow[r, "Tc"] & TT(X \otimes Y) \arrow[d, "\mu"] \\
TX \otimes TX \arrow[rr, "c"'] & & T(X \otimes Y)
\end{tikzcd}
\]
\end{definition}

\begin{remark}
\label{rem:comm_vs_symmon}
	In a symmetric monoidal category,  symmetric monoidal monads are equivalent to commutative monads, see \cite[Theorem~2.3]{Kock72} and \cite[Theorem~3.2]{Kock70}. Definition \ref{dfn: symmetric monoidal monad} corresponds to that of monad internal to the 2-category of symmetric monoidal categories, lax functors and monoidal natural transformations. The commutativity of the diagrams  
	say that $\mu$ and $\eta$ are monoidal natural transformations.
\end{remark}

\begin{definition}
A symmetric monoidal monad $T : \mathcal{C} \to \mathcal{C}$ is affine/relevant/cartesian monoidal if the underlying functor is so.
\end{definition}
 \end{document}